\documentclass[lettersize,journal]{IEEEtran}
\usepackage{amsmath,amsfonts}
\usepackage{amssymb}
\usepackage{amsthm}
\usepackage[caption=false,font=normalsize,labelfont=sf,textfont=sf]{subfig}
\usepackage{textcomp}
\usepackage{stfloats}
\usepackage{url}
\usepackage{verbatim}
\usepackage{graphicx}
\usepackage{cite}
\usepackage[ansinew]{inputenc}

\newcommand{\supp}{\mathrm{supp}}

\newcommand{\X}{\boldsymbol{X}}
\newcommand{\I}{\boldsymbol{\mathcal I}}
\newcommand{\bL}{\mathbf {\Lambda}}
\newcommand{\N}{{\mathbb N}}
\newcommand{\Z}{{\mathbb Z}}

\newcommand{\F}{{\mathbb F}}
\renewcommand{\L}{{\mathbb L}}

\newcommand{\D}{\mathcal{D}}
\newcommand{\R}{\mathcal{R}}
\renewcommand{\S}{\mathcal{S}}
\newcommand{\G}{\mathcal{G}}

\newcommand{\tq}{\;\mid\;}

\newcommand{\bs}[1]{\boldsymbol{#1}}

\newtheorem{theorem}{Theorem}[section]
\newtheorem{lemma}[theorem]{Lemma}
\newtheorem{proposition}[theorem]{Proposition}
\newtheorem{definition}[theorem]{Definition}
\newtheorem{definitions}[theorem]{Definitions}

\newtheorem{remark}[theorem]{Remark}

\newtheorem{notation}[theorem]{Notation}
\newtheorem{example}[theorem]{Example}

\newtheorem{concepto}[theorem]{}

\begin{document}

\title{Inference of unknown syndrome values in the implementation of the Berlekamp-Massey-Sakata algorithm.}

\author{Jos\'e Joaqu\'{i}n Bernal and Juan Jacobo Sim\'on,~\IEEEmembership{Member,~IEEE,}
\thanks{Manuscript received }
\thanks{Jos\'e Joaqu\'{i}n Bernal and Juan Jacobo Sim\'on are with the Departamento de Matem\'aticas, Universidad de Murcia, 30100 Murcia, Spain (e-mail: josejoaquin.bernal@um.es; jsimon@um.es)}
\thanks{Supported by Spanish government: Grant PID2020-113206GB-
I00 funded by MICIU/AEI/ 10.13039/501100011033 and 
Fundaci\'on S\'eneca of Murcia, project 22004/PI/22.}
}



\maketitle

\begin{abstract}
We study the problem of finding those missing syndrome values that are needed to implment the Berlekamp-Massey-Sakata algorithm as the Feng-Rao Majority Voting for algebraic geometric codes. We apply our results to solve syndrome correction in abelian codes.
\end{abstract}

\section{Introduction.}

The Berlekamp-Massey-Sakata algorithm (BMSa, for short) is an iterative procedure with respect to a well order in $\N\times \N$ to construct Groebner basis for the ideal of the so-called linear recurring relations of a doubly periodic array (see below for notation and definitions).

The origin of the BMSa may be found in the algorithm propossed by E. Berlekamp in 1967 \cite{berlekamp} for decoding BCH codes. In 1969, J. Massey \cite{massey} simplified the procedure and gave the termination conditions. The new version was called the Berlekamp-Massey algorithm. In 1988, Sakata extends it to two variables \cite{Sakata}, which is known as the BMSa. Finally, in 2021, in \cite{BS2}, the termination problem was completely solved, finding a minimal set of indexes that guarantees that the set of polynomials obtained at most at the last iteration in the BMSa over such a set is exactly a Groebner basis.

One of the most important applications of the BMSa may be found in locator decoding for abelian codes. Let us comment briefly how does it work (see Section~\ref{locator ideal} below for details). For any received word (formed by the sum of a codeword and an error polynomial), one may form (partially) a syndrome table (Definition~\ref{tabla sindromes}) of the error polynomial, which is a double periodic array. The idea is to find an ideal  whose defining set is exactly the set of error positions (the support of the error polynomial). To do this, from the syndrome table, we apply the BMSa to find a reduced Groebner basis for the  locator ideal and then, we may find its defining set and so, the error positions.

As we have already commented, from a received word we cannot construct all the syndrome table for the error polynomial; in fact, we only may know those values whose index corresponds to the defining set of the abelian code. So, once one has identified the minimal set of indexes needed to implement the BMSa some (few) syndrome values may be missing. In the case of algebraic geometric codes (see \cite{Cox et al Using}) there are some inference procedures (specially the Feng-Rao Majority Voting (FRMV) \cite[Theorem 10.3.7]{Cox et al Using} for one-point AG codes).

In this paper, we study the problem of finding those missing syndrome values that are needed to implment the BMSa. As the reader will see, our procedure, as the FRMV, requires previous syndrome information.

\section{Basic notions}

We shall follow all notation in \cite{BS2}. Throughout this paper $q,\,r_1,\,r_2$ will be positive integers such that $q$ is a power of a prime number, $\gcd(q,r_1r_2)=1$ and $\F=\F_q$ is a finite field of $q$-elements. As it is well-known, in this case the polynomial quotient ring $\F(r_1,r_2)=\F[X_1,X_2]/\langle X_1^{r_1}-1,X_2^{r_2}-1 \rangle$ is semi simple. We set $\I=\Z_{r_1}\times \Z_{r_2}$ where we only consider canonical representatives. 

We consider the partial ordering in $\Z\times \Z$ given by $(n_1,n_2) \preceq (m_1,m_2) \Longleftrightarrow n_1\leq m_1$ and $n_2\leq m_2.$ For $s\in \Z\times\Z$, we denote $\Sigma_s=\{(m,n)\in \Z\times \Z\tq s\preceq (m,n)\}$; in particular, setting $0=(0,0)$, we have $\Sigma_0=\Z\times\Z$. On the other hand, we will use a (total) monomial ordering \cite[Definition 2.2.1]{Cox}, denoted by ``$\leq_T$'', as in \cite[Section 2]{Sakata}. This ordering will be either the lexicographic order (with $X_1>X_2$) \cite[Definition 2.2.3]{Cox} or the (reverse) graded order (with $X_2>X_1$) \cite[Definition 2.2.6]{Cox}. The meaning of ``$\leq_T$'' will be specified as required. 

We write polynomials as it is usual:  $f \in  \F[X_1,X_2]$ is $f=f(\bs X)=\sum a_m \X^m$, where $m=(m_1,m_2)\in \Sigma_0$ and $\X^m=X_1^{m_1}\cdot X_2^{m_2}$. We write the canonical representatives $f \in  \F(r_1,r_2)$ as polynomials $f=\sum a_m \X^m$, where $m=(m_1,m_2)\in \I$. Given $f \in \F[\X]$, we denote by $\overline{f}$ its image under the canonical projection onto $\F(r_1,r_2)$, when necessary. The weight of any  $f\in \F(r_1,r_2)$ will be cardinality of the support of the polynomial; that is, $\omega(f)=|\supp(f)|$. 

For each $i\in \{ 1,2\}$, we denote by $R_{r_i}$ (resp. $\R_{r_i}$) the set of all $r_i$-th roots of unity (resp. primitive $r_i$-th roots) and we define $R=R_{r_1}\times R_{r_2}$ ($\R=\R_{r_1}\times \R_{r_2}$). We denote by $\L|\F$ a field extension that contains $R_{r_i}$, for $i=1,2$.

For any $f=f(\X) \in \F[\X]$ and $\boldsymbol{\alpha}=(\alpha_1,\alpha_2)\in R$, we write $f(\boldsymbol{\alpha})=f(\alpha_1,\alpha_2)$ and for any $m=(m_1,m_2)\in \I$, we write  $\boldsymbol{\alpha}^{m} = (\alpha_1^{m_1},\alpha_2^{m_2})$.

\begin{definition}
Let $U=(u_l)_{l\in \Sigma_0}$ be an array with entries in $\L$.We call $U$ a double periodic array, of period $r_1\times r_2$ if it satisfies the following condition:
  \begin{quotation}
  For $l=(l_1,l_2)$ and $m=(m_1,m_2)$ one has that $l_i\equiv m_i\mod r_i$ for $i=1,2$ implies that $u_l=u_m$.
  \end{quotation}
\end{definition}

\begin{definition}\label{tabla sindromes}
   Let $e\in\F(r_1,r_2)$, $\tau\in \I$ and $U=(u_n)_{n\in \Sigma_0}$ an array such that $u_n=e(\boldsymbol{\alpha}^{\tau+n})$. We shall call \textbf{syndrome values} to those $e(\boldsymbol{\alpha}^{\tau+n})$ and  $U$ will be called the \textbf{syndrome table afforded by $\tau$ and $e$}.
 \end{definition}

Note that any syndrome table as defined above is a doubly periodic array of period $r_1\times r_2$. Clearly, they entries are totally determined by its indexes over the range $0\leq n_1\leq r_1-1$ and $0\leq n_2\leq r_2-1$. In practice, working with double periodic arrays is equivalent to do it with a finite table.

As we have commented, the BMSa is an iterative procedure over the array $U$ with respect to a total ordering over $\I$; so, as in \cite[p. 269]{BS2} and \cite[p. 322]{Sakata} we specify the notion of successor. In the case of graded order we have the usual one
$$l+1=\begin{cases}
	    (l_1-1,l_2+1) & \text{if } l_1>0\\
            (l_2+1,0) & \text{if } l_1=0
      \end{cases}.$$
In the case of lexicographic order we introduce a modification
$$l+1=\begin{cases}
	    (l_1,l_2+1) & \text{if } l_2<r_2-1\\
            (l_1+1,0) & \text{if } l_2=r_2-1
      \end{cases}.$$

 As in \cite[p. 323]{Sakata}, for any $f\in \F[\X]$, we denote by $LP(f)$ the leading power product exponent, or \textbf{multidegree}, of $f$ with respect to $\leq_T$.

\begin{definition}\label{recurrencia lineal}
	Let $U$ be a doubly periodic array, $f\in \L[\mathbf X]$, with multidegree $LP(f)=s$ and $n\in \Sigma_0$. We write $f=\sum_{m\in \supp(f)}f_m\mathbf{X}^m$ and define 
	\[f[U]_n=\begin{cases}
		\displaystyle{ \sum_{m\in\supp(f)}f_m u_{m+n-s}}& \text{if }  n\in\Sigma_s\\
		0 & \text{otherwise}
	\end{cases}.\]
	
	The equality $f[U]_n=0$ will be called a linear recurring relation.\\
	
	We will say that there is a \textbf{linear recurring relation by convention} when we want to emphasize that $n\not\in\Sigma_s$; that is, there are no computation.
\end{definition}

Now, another definition related with iterations

\begin{definition}
 For $l\in \Sigma_0$, we define the finite subarray $u^l=\{u_m\tq m <_T l\text{ and } m\prec (r_1,r_2)\}$. 

Note that it is a finite set.
\end{definition}

If $u_k\in u^l$ then we also write $f[u^l]_k=f[U]_k$.

  \begin{definition}
Let $u^l\subseteq U$ be a finite subarray and  $f\in \L[\mathbf X]$ with $LP(f)=s$.
 \begin{enumerate}
 \item We say that $f$ generates  $u^l$ if $f[u^l]_k=f[U]_k=0$ on each pair $k$ such that $s\preceq k$ and $k<_T l$, and we write $f[u^l]=0$. If $\{k\in \Sigma_0\tq  s\preceq k\;\text{ and }\;k<_T l\}=\emptyset$ we will also write $f[u^l]=0$.
 
  \item We denote the \textbf{set} of generating polynomials for  $u^l$ as 
  \[\bL(u^l)=\{f\in \L[\mathbf{X}]\tq f[u^l]=0\}.\]  
  
  \item We say that $f$ generates $U$ if $f[u^l]=0$ on each pair $l\in \Sigma_0$ and we write $f[U]=0$.   
  
\item We denote the \textbf{ideal} of generating polynomials for $U$ as 
  $$\bL(U)=\left\{f\in\L[\mathbf{X}] \tq f[U]=0\right\}.$$
\end{enumerate}
  \end{definition}

\begin{notation}\label{los Delta_s}
For $s\in\Sigma_0$ we denote 
$$\Delta_s=\left\{n\in\Sigma_0\tq n\preceq s\right\}.$$

In fact, we always assume that $s\preceq (r_1-1,r_2-1)$.
\end{notation}

\begin{definitions}\label{descripcion de los Delta} Let $U$ be a doubly periodic array, $u^l\subseteq U$ and consider a set of polynomials $F_l=\left\{f^{(1)},\dots,f^{(d)}\right\}\subset\bL(u^l)$. Set $s^{(i)}=LP(f^{(i)})$, where $s^{(i)}=(s^{(i)}_1,s^{(i)}_2)$.
 \begin{enumerate}
 \item The footprint (or the connection footprint \cite[p. 1621]{Blah}) of $F$ is the set 
 \begin{eqnarray*}
\Delta(F)&=&\{n\in \Sigma_0\tq n\preceq(s^{(i)}_1-1,s^{(i+1)}_2-1)\\
&&\text{for some } 0\leq i\leq d-1\}\\
&=&\bigcup_{i=1}^{d-1}\Delta_{\left(s^{(i)}_1-1,s^{(i+1)}_2-1\right)}.
 \end{eqnarray*}
  \item We say that $F_l$ is a minimal set of polynomials for $u^l$ if
  \begin{enumerate}
   \item The sequence $s^{(1)},\dots,s^{(d)}$  (called defining points) satisfies  
   \begin{eqnarray*}
    s^{(1)}_1>\ldots>s^{(d)}_1=0 \quad\text{and}\\
    \quad 0=s^{(1)}_2<\ldots<s^{(d)}_2.
   \end{eqnarray*}
    \item If $g\in\L[\X]$ is such that $LP(g)\in \Delta(F)$ then $g\not\in\bL(u^l)$ (that is, $g[u^l]\neq 0$).
  \end{enumerate}
  \item From \cite{Blah} or \cite{Sakata} we have that, for any $l \in \Sigma_0$ and minimal sets $F,F'\subset \bL(u^l)$ we have that $\Delta(F)=\Delta(F')$, so we may write $\Delta(u^l)$ and we also call them footprints.
 \end{enumerate}
\end{definitions}

\begin{remark}\label{Conjunto de los G}
In \cite[Section 4]{Sakata}, it is proved that every set $\bL(u^l)$ has at least one minimal set of polynomials, $F_l$. Moreover, if $l,l'\in \I$ are such that $l<_T l'$ then the footprints verify $\Delta(u^l)\subseteq \Delta(u^{l'})$. 

The "corner points" $(s^{(i)}_1-1,s^{(i+1)}_2-1)$ with $i=1,\dots,d-1$ (external corner points in \cite[p. 1623]{Blah}) that determine the footprints are related with another set of polynomials that we call $G_l=\{g^{(i)}\tq i=1,\dots,d-1\}$ and pairs $k_1,\dots,k_{d-1}$ in $\I$ such that, for $i=1,\dots,d-1$, $k_i<_T l$, $g^{(i)}[u^{k_i}]_{k_i}\neq 0$ and $k_i-LP(g^{(i)})=(s^{(i)}_1-1,s^{(i+1)}_2-1)$ (see \cite[p. 327]{Sakata}). These sets are used in each iteration to update minimal sets of polynomials and, in turn, they are updated as well.
	
\end{remark}

\subsection{A brief description of the BMSa and a termination condition.}\label{descripcion del aBMS}

Following \cite[p. 331]{Sakata} (see also \cite[Paragraph III.B]{BS2}), we start with a doubly periodic array of period $r_1\times r_2$, say $U=(u_n)_{n\in \Sigma_0}$, where we have already defined a partial order $\preceq$ and a monomial ordering $\leq_T$, together with a notion of successor for $\I$.

\begin{itemize}
    \item We initialize $F_{(0,0)}=\{1\}$ and $G_{(0,0)}=\Delta(u^{(0,0)})=\emptyset$.
  \item For $l=(0,0)$, we have already the initializing objects.  
  \item Now, for $l\in \I$ with given $F_{l}$, $G_l$ and $\Delta(u^{l})$ (including $l=(0,0)$), we update them as follows:    
  \item For each $f\in F_l$ we compute $f[u^{l+1}]_l$. Then  
  \begin{enumerate}
   \item If $f[u^{l+1}]_l=0$ for all $f\in F_l$ then $F_l=F_{l+1}$,  $G_l=G_{l+1}$ and $\Delta(u^{l})=\Delta(u^{l+1})$; so that there is no strictly updating.  
   \item Otherwise, $F_l\neq F_{l+1}$.
   \begin{itemize}
    \item In this case, each $f\in F_l$ such that $f[u^{l+1}]_l\neq 0$ will be replaced following what is called \textbf{the Berlekamp procedure} \cite[Lemma 6]{Sakata} or \cite[Proposition 15]{BS2} (so that $F_l\neq F_{l+1}$).
   \end{itemize}
  \end{enumerate}  

The BMSa contemplates two alternatives in order to apply the Berlekamp procedure; namely Procedure 1 and Procedure 2 (see \cite[Paragraph III.B, pp. 271-272]{BS2} or \cite[Theorem 1 and Theorem 2]{Sakata}.   With respect to the footprint, we have
    \begin{enumerate}
  \item If $l-LP(f)\in\Delta(u^l)$, it applies Procedure 1 that gets  $\Delta(u^{l})=\Delta(u^{l+1})$ and  $G_l=G_{l+1}$ .
  
  \item If $l-LP(f)\not\in\Delta(u^l)$, it applies Procedure 2 that gets $\Delta(u^{l})\varsubsetneq \Delta(u^{l+1})$; in fact $l-LP(f)\in\Delta(u^{l+1})$. Finally, some of the replaced polynomials of $F_l$ will be used to update $G_l$ to $G_{l+1}$.
 \end{enumerate}
 \end{itemize}

\subsubsection{Sufficient conditions to find Groebner basis for $\bL(U)$.}

 In \cite{Sakata} Sakata and other authors of modern versions (see \cite{Cox et al Using,Sakata 3}) present some termination bounds. As we have commented above, in \cite{BS2} we find a minimal set of indexes (and hence the minimal number of steps) that guarantees a successful implementation of the BMSa. Let us summarize those results, that we will needed for our inference procedures.

\begin{definition}
	We define the set
	\begin{eqnarray*}
		\S(t)&=&\left\{(0,j)\in \I\tq j=0,\dots, 2t-1\right\}\\&& \cup \left\{(i,0)\in \I\tq i=0,\dots, 2t-1\right\}\\&&
		\cup \left\{(i,j)\in \I\tq i,j\neq 0\;\wedge\; i+j\leq t\right\}.
	\end{eqnarray*}
	Also,
	\begin{enumerate}
		\item We say that $\S(t)$ satisfies the l-condition (lexicographic) if $u_{(0,j)}\neq 0$, for some $j<t$.
		\item We say thay $\S(t)$ satisfies the g-condition (graded) if $u_{(i,j)}\neq 0$, for some $(i,j)$, with $i+j=1$.
	\end{enumerate}
\end{definition}

\begin{remark}\label{fuera de rango es rec lin}In \cite{BS2} the following facts are proved, for each $l\in \S(t)$ and $f\in F_l$.
 \begin{itemize}
        \item Every minimal set of polynomials may be taken to its normal form (as in \cite[p. 83]{Cox}); that is, $\supp(f)\setminus LP(f)\in \Delta(u^l)$.
        \item In any computation of $f[U]_{l}$, one has that if it happens that $l-s+m\not\in \S(t)$ then there will be a linear recurrence relation.
\end{itemize}
\end{remark}

From Theorem~23 and Theorem~28 in \cite{BS2} we have the following summary result.

\begin{theorem}\label{condiciones suficientes}
 Let $U$ be a syndrome table afforded by $\tau$ and $e$, with $\omega(e)\leq t\leq \lfloor\frac{r_i}{2}\rfloor$, for $i=1,2$. Suppose that, following the BMSa we have constructed, for $l=(l_1,l_2)\in \I$ and $u^l$, the sets $\Delta(u^l)$ and $F_l$. We have that, if $l\not\in \S(t)$ then $F_l=F_{l+1}$.
		
 Hence in order to obtain a Groebner basis for $\bL(U)$ it is enough to implement the BMSa solely on the set of indexes $\S(t)$.
\end{theorem} 

Another interesting problem, solved in \cite{BS2}, is the following one. Suppose we have an array of order $r_1\times r_2$, say $U$, and we want to see if it may be generated by linear recurrence relations. It is proved that, if $|\Delta(U)|\leq t\leq \lfloor\frac{r_i}{2}\rfloor$, for $i=1,2$, then we may found a Groebner basis for $\bL(U)$ by using only the entries $\{u_n\tq n\in \tau+\S(t)\}$ for some $\tau\in \I$ and $\S(t)$ satisfying the $l$-condition or the $g$-condition.

Once the linear recurring relations have been found, we may check if all known values of $U$ may be generated by them.

 \subsection{Locator decoding.}\label{locator ideal}

We keep all notation of the preceeding sections. As it is usual, we define an \textbf{abelian code (or bivariate code)} as an ideal $C$ in $\F(r_1,r_2)$. We denote by $d=d(C)$ its \textbf{minimum distance} and $t=\lfloor\frac{d-1}{2}\rfloor$ its \textbf{correction capability}. As it is well-known in the semisimple case, given a fixed $\boldsymbol{\alpha}\in \R$, the abelian code $C$ is totally determined by its \textbf{defining set},
 $$\D_{\boldsymbol{\alpha}}(C) = \left\{ m\in \I \tq c(\boldsymbol{\alpha}^{m})=0, \forall c\in C\right\}.$$ 
 It is also well-known that defining sets may be considered for sets of polynomials in $\F(r_1,r_2)$ (or $\L(r_1,r_2)$), and, moreover, if $B\subset \F(r_1,r_2)$ then
 $\D_{\boldsymbol{\alpha}}(B)=\D_{\boldsymbol{\alpha}}(\langle B\rangle)$.
 \smallskip
 
 Let us comment the locator decoding. Suppose a codeword $c\in C$ is sent and a polynomial $f\in \F(r_1,r_2)$ has been received. We know that $f=c+e$; where $e\in\F(r_1,r_2)$ is, as usual, the error polynomial which we want to determine. This will be done in two steps: the first one is to find the error locations; that is, $\supp(e)$. The second one is to determine its coefficients. The last step may be done directly by solving a system of linear equations with coeficients in $\F$. So, we shall concentrate at the first step.
 \smallskip
 
 To find the error positions or locations, we first construct the partial syndrome table of $e$. Clearly, for any $m\in \D_{\bs{\alpha}}(C)$ we have that $f(\bs\alpha^m)=e(\bs\alpha^m)$. Now, if there is $\tau\in\I$ such that $\S(t)+\tau\subset \D_{\bs{\alpha}}(C)$ then we may form a partial syndrome table $U$ afforded by $e$ and $\tau$ and we may apply the BMSa and Theorem~\ref{condiciones suficientes} claim that we get a Groebner basis for $\Lambda(U)$.
 \smallskip
 
Secondly, we find the error locations. To do this, we consider the following ideal.
 
  \begin{definition}{\cite[Definition 18]{BS2}}
 Let $e \in \F(r_1,r_2$. The locator ideal (with respect to $\bs\alpha\in \R$) is
  \[L(e)=\left\{f\in \L(r_1,r_2)\tq f(\boldsymbol{\alpha}^n)=0,\;\forall n\in \supp(e)\right\}.\]
 \end{definition}
 
 Note that, as $\L$ is a splitting field, we have that $\D_{\bs{\alpha}}(L(e))=\supp(e)$. There is an important relation between defining sets and footprints (in the classical sense for ideals \cite{Cox}) that we may find in \cite[Remark 7]{BS2} (together with other interesting properties), to witt
 \begin{equation}\label{cardinales de delta y conj definicion}
  |\Delta(L(e))|=|\D_{\boldsymbol{\alpha}}(L(e))|.
 \end{equation}
 
Now let $U$ be a syndrome table afforded by $\tau$ and $e$. We denote by $\overline{\bL(U)}$ the canonical projection of $\bL(U)\leq \L[\X]$ onto the ring $\L(r_1,r_2)$. Next result explains us the relationship between the ideals $\bL(U)$ and $L(e)$.

\begin{theorem}{\cite[Equation 3]{BS2}}\label{el locator y el delta}
In the setting above, we have
 \[\overline{\bL(U)}=L(e).\]

Hence 
 $$\D_{\boldsymbol{\alpha}}(\overline{\bL(U)})=\D_{\boldsymbol{\alpha}}(L(e))=\supp(e).$$
\end{theorem}

Here, it appears the interest of the BMSa. If we find a Groebner basis, say $\G$, for $\Lambda(U)$ then we have that $\D_{\bs\alpha}(\G)=\supp(e)$, the error locations.\\

We finally comment a problem that we may have when we are constructing the partial syndrome table and whose study is the interest of this paper. Let us explain it by an example.

\begin{example}\label{ejemplo planteamiento}
 We set $q=2$, $r_1=r_2=15$, so that $\L=\F_{2^4}$, with primitive root $a\in \L$. Let $C$ be the abelian code in $\F_2(15,15)$ with defining set $\D_{\bs \alpha}(C)=Q(0,0)\cup Q(0,1)\cup Q(0,3)\cup Q(0,5)\cup Q(1,0)\cup Q(3,0)\cup Q(5,0)\cup Q(1,1)\cup Q(2,1)$, with respect to some $\bs\alpha\in \R$, where $Q(i,j)$ denotes the $2$-orbit of $(i,j)$ modulo $(r_1,r_2)$ (see \cite[Definition 38]{BS2} and the paragraph below it).
 
 By results in \cite{BBCS2}, the strong apparent distance of $C$ is $sd^*(C)=8$, so we set $t=3$. Suppose that a polynomial $f\in \F_2(15,15)$ was received. Then we compute the syndrome values over $\D_{\bs \alpha}(C)$. We write for $\tau=(0,0)$ the values corresponding to $\S(3)$.
 
  \[U=\left(\begin{array}{lllllllllllllll}
   1 & a^6 & a^{12} & a^{9} & a^{9}& 0\\
   
   a^{12} & a^{3} & u_{(1,2)}  \\
   
   a^{9} & a^{8}  \\
   
   a^{7} \\
   
   a^{3}\\
   a^5
  \end{array}\right)\]
  
The value $u_{(1,2)}$ is missing because $(1,2)\not\in \D_{\bs\alpha}(C)$. Hence we cannot implement the algorithm through the entire set $\S(3)$.
\end{example}

As we commented in Introduction, we want to find alternative procedures to infer such values or to get somehow the reduced Groebner basis for $\Lambda(U)$.

\section{Inference of values.}

\begin{concepto}[General Assumption.]\label{Situac general}
 Throughout this section, we assume that $\F$, $\L$, $t,r_1,r_2,\tau$ and $e$ are defined as in previous sections, with  $\omega(e)\leq t\leq \lfloor\frac{r_i}{2}\rfloor$, for $i=1,2$. Let $U$ be the syndrome table afforded by $e$ and $\tau$.

 We suppose that following the implementation of the BMSa (under the lexicographic ordering with $\S(t)$ satisfying the l-condition or the graduate ordering, with $\S(t)$ satisfying the g-condition) we have constructed  for $l=(l_1,l_2)\in \S(t)$, the sets $\Delta(u^l)$ and $F_l=\left\{f^{(1)},\dots,f^{(d)}\right\}\subset\bL(u^l)$. We set $s^{(i)}=LP(f^{(i)})$ and $s^{(i)}=(s^{(i)}_1,s^{(i)}_2)$, for $i=1,\dots,d$.\smallskip
 
 We want to compute the step $l$ to obtain $\Delta(u^{l+1})$ and $F_{l+1}$, but, we do not know the value of the entry $u_{l}$.
\end{concepto}

 So, we cannot compute directly $f[U]_l$ for every $f\in F_l$, satisfying $LP(f)\preceq l$. We need an alternative way to estimate such value or to get the Groebner basis for $\Lambda(U)$.
 
 As one may expect, we cannot estimate \textit{any} value of $U$; in fact we have to restrict ourselves to unknown values having ``border indexes''; that is, the pairs $l=(l_1,l_2)\in \S(t)$ such that:
 
 \begin{enumerate}
 \item If $l_1,l_2\neq 0$ then $l_1+l_2=t$.
 \item \begin{enumerate}
        \item If $l_1=0$ then $t+s_2^{(d)}-1\leq l_2\leq 2t-1$.
        \item If $l_2=0$ then $t+s_1^{(1)}-1\leq l_1\leq 2t-1$.
       \end{enumerate}
\end{enumerate}

The pairs of the first class are clearly border points. To understand the criteria for the second class, the reader may see Lemma~\ref{inferencia en los ejes general} below.

We begin with two technical results.

\begin{lemma}\label{area rectangulo Delta}
 For $(a,b)\in \Sigma_0$, with $a,b\neq 0$, we have that $|\Delta_{(a-1,b-1)}|=ab$.
\end{lemma}
\begin{proof}
 Straightforward.
\end{proof}

\begin{lemma}\label{primer paso no cero}
We set $\F$, $\L$, $t,r_1,r_2,\tau$ and $e$, as above, with  $\omega(e)\leq t\leq \lfloor\frac{r_i}{2}\rfloor$, for $i=1,2$. Let $U$ be the syndrome table afforded by $e$ and $\tau$, having the $l$-condition or the $g$-condition as required. If $l=(l_1,l_2)\in \Sigma_0$ is the first element for which $u_l\neq 0$ then
 \begin{enumerate}
  \item  Under the lexicographic ordering, $F_{l+1}=\{X_1,X_2^{l_2+1}\}$.
  \item Under the graduate ordering $F_{l+1}=\{X_1^{l_1+1},X_2^{l_2+1}\}$
 \end{enumerate}
\end{lemma}
\begin{proof}
 \textit{1.} As the $l$-condition is satisfied then $l_1=0$, $F_l=\{1\}$ and $\Delta(u^l)=\emptyset$. Then, $1[U]_l=u_l\neq 0$ and as $s^{(1)}=(0,0)$, then $l-s^{(1)}\notin \Delta(u^l)$; so that, we apply Procedure 2 (see Paragraph~\ref{descripcion del aBMS}).
 
By analizing the five cases of constructions considered in Procedure 2 (see \cite[pp. 271-272]{BS2}) one may check that it is only possible to apply cases (3) and (5), because all others need that $d\geq 2$. The mentioned cases get $F_{l+1}=\{X_ 1,X_2^{l_2+1}\}$, $G=\{1\}$ and $\Delta(u^{l+1})=\{0\}$.
 
 \textit{2.} In this case, one may check that, again, it is only possible to apply the constructions considered in cases (3) and (5), by the same reason. As $l \in \{(0,0),\,(1,0),\,(0,1)\}$ we get the result.
\end{proof}

\begin{lemma}\label{d no pasa de 4}
In the setting of General Assumption~\ref{Situac general}, suppose that $l_1,l_2\neq 0$ and $l_1+l_2=t$. If $d\geq 4$ then for $i\in \{1,d\}$ there is a polynomial $f^{(i)}\in F_l$ such that $LP(f^{(i)})\preceq l$ and  $f^{(i)}[U]_l=0$.
\end{lemma}
\begin{proof}
Set $\Delta=\Delta(u^{l+1})$. We recall that $|\Delta|\leq t$. First, note that, as $s_1^{(1)}+s_2^{(d)}-1\leq t=l_1+l_2$ then $LP(f^{(1)}) \preceq l$ or $LP(f^{(d)}) \preceq l$, so that there is at least one element $i\in \{1,d\}$ for which $LP(f^{(i)})\preceq l$.

Suppose that $LP(f^{(1)})\preceq l$. If $f^{(1)}[U]_l=0$ the we are done. Otherwise, having in mind that $l-s^{(1)}\in\Delta$ and $l_2\neq 0$, we have that the set $A=\{(0,l_2),\dots,(l_1-s_1^{(1)},l_2)\}\cup\{(l_1-s_1^{(1)},0),\dots,(l_1-s_1^{(1)},l_2-1)\}\cup \{(x,0)\tq x\in \{0,\dots,s_1^{(1)}-1\}\setminus\{(l_1-s_1^{(1)},0)\}$ verifies that $A\subseteq \Delta$ and (as it is disjoint union) $|A|=l_1-s_1^{(1)}+1+l_2+s_1^{(1)}-1=t$. So that $A=\Delta$.

Now, since $d\geq 4$ then $s_2^{(d)}\geq 3$ and so, if $l_2=1$ then the point $(0,l_2+1)\in \Delta\setminus A$, a contradiction. If $l_2\geq 2$ then, if $l_1\neq s_1^{(1)}$, the point $(0,1)\in \Delta\setminus A$; in case $l_1 = s_1^{(1)}$, the point $(1,1)\in \Delta\setminus A$. Both cases drive us to a contradiction 

The case $LP(f^{(d)})\preceq l$ is completely analogous.
\end{proof}

\begin{theorem}\label{Caso l1 y l2 mayores que 1}
In the setting of General Assumption~\ref{Situac general} suppose that $l_1+l_2=t$. If $l_1,l_2>1$  then there exist $f\in F_l$ such that $LP(f)\preceq l$ for which $f[U]_l=0$.
\end{theorem}
\begin{proof}
 By Lemma~\ref{d no pasa de 4} and since $\S(t)$ satisfies the $l$-condition or the $g$-condition, we may assume that $2\leq d\leq 3$.
  Set $\Delta=\Delta(u^{l+1})$. We shall proceed by considering cases. Note that  $l_1,l_2\neq 0$ implies $t\geq 2$; besides, it always happen that $s_1^{(1)}+s_2^{(d)}-1\leq |\Delta|$.\\
  
   \textbf{CASE 1.} Suppose that $l_1< s_1^{(1)}$. Then we have that $l_2\geq s_2^{(d)}$. It may happen that $f^{(i)}[U]_l\neq 0$ for $1<i\leq d$ (note that $f^{(1)}[U]_l=0$ by convention in Definition~\ref{recurrencia lineal}). 
 
\textbf{Case 1(A).} If $l_2=s_2^{(d)}$ then $l_1+1=s_1^{(1)}$ and so  $s_1^{(1)}+s_2^{(d)}=t+1$ (that is, there are already $t$ points on the axes) and $2\leq d\leq 3$.
 
\textbf{Case 1(A)(I).} If $d=2$ then, having in mind that $\Delta(u^l)$ must be a rectangle, we have that $s_2^{(d)}=1=l_2$, because we already have $t$ points on the axes and $s_1^{(1)}\geq 2$. But this contradicts our assumption that $l_1,l_2>1$.

\textbf{Case 1(A)(II).} Suppose that $d=3$; so that $s^{(2)}=(1,1)$ must be a defining point (otherwise $|\Delta(u^l)|>t$) and $s^{(2)}\preceq l$. If $f^{(2)}[U]_l\neq 0$ then $(l_1-1,l_2-1)\in\Delta$ (recall that $\Delta=\Delta(u^{l+1})$), which already have $t$-points on the axes, so we must have that $l_1=1$ or $l_2=1$ (this last equality cannot happen because $d=3$ and we are assuming that $l_2=s_2^{(d)}$). This, again, contradicts our assumption that $l_1,l_2>1$.

\textbf{Case 1(B).} Suppose that $l_2>s_2^{(d)}$; so that, $s^{(d)}\preceq l$. If $f^{(d)}[U]_l\neq 0$, then $(l_1,l_2-s_2^{(d)})\in\Delta$ and it does not belong to any of the axes. Then, we have that $t\geq |\Delta|\geq s_1^{(1)}+s_2^{(d)}-1+l_1+l_2-s_2^{(d)}-1=t+s^{(1)}_1-2>t$, because $s^{(1)}_1 > 2$, which is impossible. Hence $f^{(d)}[U]_l = 0$.

\textbf{CASE 2.} Suppose that $l_1= s_1^{(1)}$; so that $LP(f^{(1)})\preceq l$.
 
\textbf{Case 2(A).} If $l_2 < s_2^{(d)}$ then we must have that $s_1^{(1)}+s_2^{(d)}-1\leq t < s_1^{(1)}+s_2^{(d)}$, from which $s_1^{(1)}+s_2^{(d)}=t+1$.
 
\textbf{Case 2(A)(I).} If $d=2$ then $t\geq |\Delta|=s_1^{(1)}s_2^{(2)}$, from which $s_1^{(1)}s_2^{(2)}< s_1^{(1)}+s_2^{(2)}$. This only happens for the values $s_1^{(1)}=1$ or $s_2^{(2)}=1$, which is not the case. So $d\neq 2$.

\textbf{Case 2(A)(II).} Now, suppose that $d=3$. As $s_1^{(1)}+s_2^{(3)}=t+1$ then $s^{(2)}=(1,1)$. If $f^{(2)}[U]_l\neq 0$ then $l-(1,1)\in \Delta$ and it does not belong to the axes; so that $|\Delta|>t$, which is impossible. Hence $f^{(2)}[U]_l=0$.
 
 \textbf{Case 2(B).} Suppose that $l_2 \geq s_2^{(d)}$. If $f^{(1)}[U]_l\neq 0$ then $l-s^{(1)}=(0,l_2)\in \Delta$; so that $|\Delta|=t$, counting only its points on the axes. If it also happen that $f^{(d)}[U]\neq 0$ then $l-s^{(d)}=(l_1,l_2-s_2^{(d)})\in \Delta$, and it is not a point of $\Delta(u^{(l)})$, which is impossible. Hence $f^{(1)}[U]_l = 0$ or $f^{(d)}[U]_l= 0$.

\textbf{CASE 3.} Suppose that $l_1> s_1^{(1)}$. If $f^{(1)}[U]_l\neq 0$ then $(l_1-s_1^{(1)},l_2)\in \Delta$ from which $|\Delta|\geq s_1^{(1)}+s_2^{(d)}-1+l_1-s_1^{(1)}+l_2-1=t+s_2^{(d)}-2$, so we only have to consider the values $s_2^{(d)}\leq 2$.

Since $l_2 < s_2^{(d)}$ implies $l_2=1$, we only consider the case in which $l_2\geq s_2^{(d)}$. As $f^{(1)}[U]_l\neq 0$ then, if $2\leq l_2=s_2^{(d)}$ then $s^{(d)}\in \Delta$, and so $|\Delta|=t+s_2^{(d)}-1>t$, which is impossible. If $l_2>s_2^{(d)}$ then $t\geq |\Delta|\geq l_1-s^{(1)}_1+l_2+s^{(d)}_2+s^{(1)}_1\geq t+1$, which is impossible too. Hence $f^{(1)}[U]_l = 0$.
\end{proof}

\begin{remark}
 Following the proof of the theorem above, we may specify which element or elements $f\in F_l$ satisfying $LP(f)\preceq l$ may give us a linear recurring relation.
 \begin{enumerate}
  \item If $l_1 \leq s_1^{(1)}$ and $l_2\leq s_2^{(d)}$ then $d=3$ and $f^{(2)}[U]_l=0$.
  \item If $l_1<s_1^{(1)}$ and $l_2>s_2^{(d)}$ then $f^{(d)}[U]_l=0$.
  \item If $l_1=s_1^{(1)}$ and $l_2 \geq s_2^{(d)}$ then $f^{(1)}[U]_l=0$ or $f^{(d)}[U]_l=0$.
  \item If $l_1 > s_1^{(1)}$ then $f^{(1)}[U]_l=0$.
 \end{enumerate}

\end{remark}

\begin{theorem}\label{inferencia principal}
In the setting of General Assumption~\ref{Situac general} we have:

\begin{enumerate}
 \item If $l=(1,t-1)$ then there exist $f\in F_l$ such that $LP(f)\preceq l$ for which $f[U]_l=0$, except in the following cases:
 \begin{enumerate}
   \item $d=2$, $s_1^{(1)}=1$, $s_2^{(2)}=t$
    and the implementation is doing under the lexicographic ordering ($X_1>X_2$).
   
  \item $d=2$, $s^{(1)}=(2,0)$, $2 \leq t= 2s_2^{(2)}$.
  
  If the implementation is doing under the lexicographic ordering then $l\neq (1,1)$.
   
   \item $d=3$, $s^{(1)}=(2,0)$, $t= s_2^{(2)}+s_2^{(3)}$.
 \end{enumerate}
  \item If $l=(t-1,1)$ and $t>2$ then there exist $f\in F_l$ such that $LP(f)\preceq l$ for which $f[U]_l=0$, except in the following cases:
  \begin{enumerate}
  \item $d=2$, $s_1^{(1)}=t$, $s_2^{(2)}=1$
   and the implementation is doing under the inverse graded ordering ($X_2>X_1$).  
  \item $d=2$, $s^{(2)}=(0,2)$, $t= 2s_1^{(1)}$.
  
  \item $d=3$, $s^{(3)}=(0,2)$, $t=s_1^{(1)}+s_1^{(2)}$.
 \end{enumerate}
\end{enumerate} 
\end{theorem}

\begin{proof} Again, by Lemma~\ref{d no pasa de 4} and since $\S(t)$ satisfies the $l$-condition or the $g$-condition, we may assume that $2\leq d\leq 3$.
  Set $\Delta=\Delta(u^{l+1})$. We shall proceed by considering cases. Note that  $l_1,l_2\neq 0$ implies $t\geq 2$; besides, it always happen that $s_1^{(1)}+s_2^{(d)}-1\leq |\Delta|$.
  
  We shall only proof the theorem for $l=(1,t-1)$. The proof for the other pair is analogous. We proceed by analyzing the first two cases considered in the proof of Theorem~\ref{Caso l1 y l2 mayores que 1}.\\
 
 \textbf{CASE 1.} Suppose that $1=l_1< s_1^{(1)}$. Then we have that $t-1\geq s_2^{(d)}$ because $s_1^{(1)}+s_2^{(d)}\leq t+1$. In this case, $f^{(1)}[U]_l=0$ by convention in Definition~\ref{recurrencia lineal}. 
 
 
\textbf{Case 1(A).} If $s_2^{(d)}=t-1=l_2$ then as $s^{(1)}_1+t-1 \leq t+1$ then  $s_1^{(1)}\leq 2$, so that $s^{(1)}_1=2$, $2\leq d\leq 3$ and $s_1^{(1)}+s_2^{(d)} = t+1$; that is, there are already $t$ points of $\Delta(u^l)$ on the axes.
 
\textbf{Case 1(A)(I).} If $d=2$ then, having in mind that $\Delta(u^l)$ must be a rectangle, we have that $s_2^{(2)}=1=l_2=t-1$, because we already have $t$ points on the axes and $s_1^{(1)}= 2$ (note that $s^{(1)}_1s^{(2)}_2<s^{(1)}_1+s^{(2)}_2$); so that $t=2$ and $l=(1,1)$. If $f^{(2)}[U]_{(1,1)}\neq 0$ then, as $(1,1)-s^{(2)}=(1,0)\in \Delta(u^l)$ there are no contribution with a new point to $\Delta$.

Later, we shall see that the values $s^{(2)}=(0,1)$, $l=(1,1)$ and $d=2$ cannot be reached when one is implementing the BMSa under the lexicographic ordering.

In case we are implementing the BMSa under the graded ordering we will see that there are not any problem so \textbf{we cannot guarantee that} $f^{(2)}[U]_{(1,1)}=0$. This case is \textbf{Item \textit{1(b)} of this theorem} for $t=2$.

\textbf{Case 1(A)(II).} Suppose that $d=3$; so that $s^{(2)}=(1,1)$ must be a defining point (otherwise $|\Delta(u^l)|>t$) and then $s^{(2)}\preceq l$. Even $f^{(2)}[U]_l\neq 0$, as $l-s^{(2)}=(0,t-2)\in\Delta(u^l)$, because $s^{(3)}=(0,t-1)$, there are no any problem or contradiction. Also, if $f^{(3)}[U]_l\neq 0$ then $l-s^{(3)}=(1,0)\in \Delta(u^l)$.

So, for  $s^{(1)}=(2,0)$,  $s^{(2)}=(1,1)$, $s^{(3)}=(0,t-1)$ and $l=(1,t-1)$ under any of the orders considered \textbf{we cannot guarantee that} $f^{(i)}[U]_l=0$, for any of $i=2,3$. This case will be in another one that we will see later to form the \textbf{Item \textit{1(c)} of this theorem}.
 
\textbf{Case 1(B).} Suppose that $s_2^{(d)}<l_2=t-1$ (Note that $t>2$). If $f^{(d)}[U]_l\neq 0$ then $l-s^{(d)}=(1,t-1-s_2^{(d)})\in\Delta$ and it does not belong to any of the axes. If, moreover, $t-1\geq 2s_2^{(d)}$ then the point $s^{(d)}$ also belong to $\Delta$ so we have that $t\geq |\Delta|\geq s_1^{(1)}+s_2^{(d)}+t-1-s_2^{(d)}=t+s^{(1)}_1-1>t$, because $s^{(1)}_1\geq 2$, which is impossible. So, it must happen that $t-1< 2s_2^{(d)}$ and so, as it may be $s^{(d)}\not\in\Delta$ then $t\geq |\Delta|\geq t+s^{(1)}_1-2$, which forces $s^{(1)}_1=2$.

\textbf{Case 1(B)(I).} If $d=2$, we have that $t-1<2s_2^{(2)}\leq t$ (since $s_1^{(1)}>1$, then $2s^{(2)}=|\Delta(u^l)|\leq t$), from which the equality holds, and $s^{(1)}_1=2$ by the paragraph above. Then $l-s^{(2)}=(1,\frac{t}{2}-1)\in\Delta(u^l)$ and so, \textbf{we cannot guarantee that} $f^{(2)}[U]_l=0$.  This is \textbf{Item \textit{1(b)} of this theorem for $t>2$}.
 
\textbf{Case 1(B)(II).} Let $d=3$. If $l_2\geq s_2^{(2)}+s_2^{(3)}$ and $f^{(i)}[U]_l\neq 0$ for $i=2,3$ (note that $s^{(i)}\preceq l$) again, the point $s^{(3)}$ becomes an element of $\Delta$, so, as $(1,t-1-s^{(3)}_2),(0,s^{(3)}_2)\in \Delta$, we have that $|\Delta|\geq s^{(3)}_2+1+t-1-s^{(3)}_2+1=t+1$, which is impossible; so that $f^{(i)}[U]_l= 0$ for some $i\in\{2,3\}$.

Otherwise, suppose that  $l_2< s_2^{(2)}+s_2^{(3)}$ then we do not have to add more points to $\Delta$, so  $|\Delta|\leq t$. Hence, \textbf{we cannot guarantee that} $f^{(i)}[U]_l=0$, for any of $i=2,3$.

We note that this case contains that considered in \textbf{Case 1(A)(II)} (except for the fact that, before, $s^{(2)}=(1,1)$) to form the \textbf{Item \textit{1(c)} of this theorem}, including $l_2=s_2^{(3)}$. In this case, we have that $t-1< s_2^{(2)}+s_2^{(3)}\leq t$, from which the equality holds.\\
 
 \textbf{CASE 2.} Suppose that $l_1= s_1^{(1)}$; so that, $s^{(1)}_1=1$ and so $d=2$.
 
\textbf{Case 2(A).} If $s_2^{(2)}=t$ then $|\Delta(u^l)|=t$ and $f^{(2)}[U]_l= 0$ by convention in Definition~\ref{recurrencia lineal}. If $f^{(1)}[U]_l\neq 0$ we have that $l-s^{(1)}\in\Delta(u^l)$ and, as this does not imply any increasing of $|\Delta|$. As above, later, we shall see that the values $s^{(1)}=(1,0)$ (so $d=2$), $s^{(2)}=(0,t)$ and $l=(1,t-1)$ cannot be reached when one is implementing the BMSa under the graded ordering.

So, \textbf{we cannot guarantee that} $f^{(1)}[U]_l=0$, under the lexicographic ordering. This is \textbf{Item \textit{1(a)} of this theorem}. 

\textbf{Case 2(B).} Suppose that $s_2^{(2)} \leq t-1$ and $f^{(i)}[U]_l\neq 0$ for $i=1,2$. Then $l-s^{(1)}=(0,t-1)\in\Delta$ and so, it has already $t$ points by counting only the points in the axes; however, $l-s^{(2)}=(1,l_2-s_2^{(d)})\in \Delta$, and it is not a point in any of the axes, giving us a contradiction. Hence $f^{(1)}[U]_l= 0$ or $f^{(2)}[U]_l= 0$.\\

We begin by seeing that the values of the parameters involved in  \textit{Statement 1(a) of this theorem} cannot be reached when the implementation of the BMSa is doing under the inverse graded ordering. Let $s^{(1)}=(1,0)$ (so $d=2$), $s^{(2)}=(0,t)$ and $l=(1,t-1)$. By the description of the BMSa \cite[Algorithm 1, p. 272]{BS2}, we know that there is the polynomial $g^{(1)}$ and a pair $l'<_Tl$ such that $l'-LP(g^{(1)})=(0,t-1)$ which is a corner point (Remark~\ref{Conjunto de los G}). This means that $l'_2-LP(g^{(1)})_2=t-1$. If $LP(g^{(1)})_2\geq 1$, by Lemma~\ref{primer paso no cero} we have $(0,t)\preceq l'$, which is impossible because it is the largest point on the diagonal, so that $l<_T l'$. If $LP(g^{(1)})_2=0$ then $LP(g^{(1)})_1 \geq 1$; so that $l'_1\geq 1$ and $l'_2=t-1$ contradicting the fact that $l'<_T l$.  

Now we are going to see that the values of the parameters involved in  \textit{Statement 1(b) of this theorem, with $l=(1,1)$} cannot be reached when the implementation of the BMSa is doing under the lexicographic ordering. As $t=2$ and so $s_2^{(d)}=1$ then $l=(1,1)$. One may check that it must be $F_{(1,1)}=\{X_1+b_1,X_2+b_2\}$, for some $b_1,b_2\in \L$; so that $l_1\nless s_1^{(1)}$. (We note that in case we are implementing the BMSa under the graded ordering, it must be $(2,0)<_T (1,1)=l$, and this does not get any contradiction.)
 
As we commented at the beginning of this proof, the argument for $l=(t-1,1)$ is analogous.
\end{proof}

\begin{remark}
\begin{enumerate}
 \item By the proof of the theorem above, one may see that, when the exception conditions do not hold, it is not possible to specify which element or elements $f\in F_l$ satisfying $LP(f)\preceq l$ could give us a linear recurring relation. We only know that there will be at most two elements of $F_l$.
 
 \item In the proof of the theorem above, we use the phrase ``we cannot guarantee that...''; however, in all cases, we have examples of tables (in which we know all values) where, there are no recurring relations at ``exception steps''. One of them is the following example.
\end{enumerate} 
\end{remark}

\begin{example}\label{esquivando el valor perdido con el orden}

Let us go back to Example~\ref{ejemplo planteamiento}; so that
   $q=2$, $r_1=r_2=15$, $\L=\F_{2^4}$, with primitive root $a\in \L$, and $t=3$. Our table is
  \[\left(\begin{array}{lllllllllllllll}
   1 & a^6 & a^{12} & a^{9} & a^{9}& 0\\
   
   a^{12} & a^{3} & u_{(1,2) } \\
   
   a^{9} & a^{8}  \\
   
   a^{7} \\
   
   a^{3}\\
   a^5
  \end{array}\right)\]
  
Then, if we implement the BMSa under the lexicographic ordering we have (summarizing)
\begin{footnotesize}
 \[\begin{array}{|l|l|l|l|}\hline
 l&F\subset \bL(u^{l+1})&G&\Delta(u^{l+1})\\ \hline
  \text{Initiation}&\{1\}&\emptyset&\emptyset\\ \hline
(0,0)\rightarrow&\{X_1,X_2\}&\{1\}&\{(0,0)\}\\ \hline
(0,1)\rightarrow&\{X_1,X_2+a^{6}\}&\{1\}&\{(0,0)\}\\ \hline
 (0,2)\rightarrow&\text{The same}&\text{The same}&\text{The same}\\ \hline
 (0,3)\rightarrow&\{X_1,X_2^3+a^{6}X_2^2+a\}&\{X_2+a^6\}&\begin{array}{l}
                   \{(0,0),\\(0,1),\\
                    (0,2)\}
                   \end{array}\\ \hline
  (0,4)\rightarrow&\text{The same}&\text{The same}&\text{The same}\\ \hline
   (0,5)\rightarrow&\begin{array}{l}
                    \{X_1,\\
                    X_2^3+a^{6}X_2^2+a^5X_2+a^6\}
                   \end{array}&\{X_2+a^6\}&
                   \begin{array}{l}
                   \{(0,0),\\(0,1),\\
                    (0,2)\}
                   \end{array}
\\ \hline
 (1,0)\rightarrow& \begin{array}{l}
                    \{X_1+aX_2+a^2,\\
                    X_2^3+a^{6}X_2^2+a^5X_2+a^6\}
                   \end{array}&\{X_2+a^6\}&
                   \begin{array}{l}
                   \{(0,0),\\(0,1),\\
                    (0,2)\}
                   \end{array}\\ \hline
(1,1)\rightarrow&\text{The same}&\text{The same}&\text{The same}\\ \hline
\end{array}\]
\end{footnotesize}

So, at the step $l=(1,2)$ we are in the situation of Theorem~\ref{inferencia principal}.\textit{1(a)}. Then, we try with the reverse graded ordering. (In fact, we know that $u_{(1,2)}=0$ so that $f[U]_{(1,2)}=a^{11}$.)

\begin{footnotesize}
 \[\begin{array}{|l|l|l|l|l|}\hline
 l&F\subset \bL(u^{l+1})&G&\Delta(u^{l+1})\\ \hline
  \text{Initiation}&\{1\}&\emptyset&\emptyset\\ \hline
(0,0)\rightarrow&\{X_1,X_2\}&\{1\}&\{(0,0)\}\\ \hline
(1,0)\rightarrow&\{X_1+a^{12},X_2\}&\{1\}&\{(0,0)\}\\ \hline
 (0,1)\rightarrow&\{X_1+a^{12},X_2+a^6\}&\{1\}&\{(0,0)\}\\ \hline
 \begin{array}{l}
  (2,0)\\
  (1,1)\\
  (0,2)
 \end{array}
\rightarrow&\text{The same}&\text{The same}&\text{The same}\\ \hline
(3,0)\rightarrow&\begin{array}{l}
                    \{X_1^3+a^{12}X_1^2+a^{10},\\
                    X_2+a^6\}
                   \end{array}&\{X_1+a^{12}\}&
                   \begin{array}{l}
                   \{(0,0),\\(1,0),\\
                    (2,0)\}
                   \end{array}\\ \hline
 (2,1)\rightarrow& \begin{array}{l}
                    \{X_1^3+a^{12}X_1^2+a^{10},\\
                    X_2+a^7X_1+a^{12}\}
                   \end{array}&\{X_1+a^{12}\}&
                    \begin{array}{l}
                   \{(0,0),\\(1,0),\\
                    (2,0)\}
                   \end{array}\\ \hline
\end{array}\]
\end{footnotesize}

Now, at step $l=(1,2)$, Theorem~\ref{inferencia principal} tell us that $f^{(2)}[U]_{(1,2)}= 0$. Then $0=u_{(1,2)}+a^7u_{(2,1)}+a^{12}u_{(1,1)}=u_{(1,2)}$ and we get the desired value.
\end{example}

Of course, changing orderings in the implementation of the BMSa does not always work. The proposition below tell us how we may proceed to find the minimal Groebner basis for $\bL(U)$, when such change fails. 

We recall that the BMSa contemplates two alternatives in order to apply the Berlekamp procedure; namely Procedure 1 and Procedure 2 (see \cite[Paragraph III.B, pp. 271-272]{BS2} or \cite[Theorem 1 and Theorem 2]{Sakata}. 

In Proposition~\ref{salvando los muebles para los casos del Teorema con lex}, Proposition~\ref{salvando los muebles para los casos del Teorema con Grad} and Proposition~\ref{salvando los muebles en los ejes} we mention some unknown values (denoted by $b$ and $c$). We refer the reader to Section~\ref{final procedure}where we explain how to face the computation of these values.

\begin{proposition}\label{salvando los muebles para los casos del Teorema con lex}
Take the setting of General Assumption~\ref{Situac general}, with $F_l=\{f^{(1)},\dots,f^{(d)}\}$, and let $g^{(i)}\in G_l$, as in Remark~\ref{Conjunto de los G} such that $g^{(i)}[u^{k_i}]_{k_i}=v_i\neq 0$, with $k_i<_Tl$, for $i=1,\dots,d-1$. By Remark~\ref{fuera de rango es rec lin}, we may suppose that $\supp(f^{(i)})\setminus LP(f^{(i)})\in \Delta(u^l)$.  Suppose that we are implementing the BMSa under the lexicographic ordering ($X_1>X_2$).

\begin{enumerate}
 
 \item Consider $l=(1,t-1)$. 

 \begin{enumerate}
 
 \item Suppose we have that $d=2$, $s^{(1)}=(1,0)$ and $s^{(2)}=(0,t)$.  Then there is a value $b\in\L$ such that the polynomial 
 $$h_b=f^{(1)}-\frac{b}{v}g^{(1)}$$
 verifies that $\G_b=\{h_b, f^{(2)}\}$ is a Groebner basis for $\bL(U)$.

\item Suppose that following the BMSa we have that $d=2$, $s^{(1)}=(2,0)$, $t= 2s_2^{(2)}$. Then there is a value $b\in\L$ such that the polynomial 
$$h_b=f^{(2)}-\frac{b}{v}g^{(1)}$$
verifies that $\G_b=\{\overline f^{(1)}, h_b\}$ is a Groebner basis for $\bL(U)$, where $\overline f^{(1)}$ is the final updating of $f^{(1)}$ through the steps $(2,0),\dots,(2,\frac{t}{2}-1),(3,0),\dots,(3,\frac{t}{2}-1)$.

\item Suppose that following the BMSa we have that $d=3$, $s^{(1)}=(2,0)$, $t= s_2^{(2)}+s_2^{(3)}$.

There exists a list of values $c,b,b_0,\dots,b_{s_2^{(2)}-1}$ such that, setting $\mathbf{b}=(b,b_0,\dots,b_{s_2^{(2)}-1})$,
\[h_{\mathbf b}=f^{(2)}-\frac{b}{v_2}g^{(2)}-\sum_{n=0}^{s^{(2)}_2-1} \frac{b_n}{v_1}X_2^{s^{(2)}_2-(n+1)}g^{(1)} \]
and 
\[h_c=f^{(3)}-\frac{c}{v_1}g^{(1)}\]
verify that $\G_{(\mathbf b,c)}=\{\overline f^{(1)},\; h_{\mathbf b}, h_c\}$ is a Groebner basis for $\bL(U)$, where $\overline f^{(1)}$ is the final updating of $f^{(1)}$ through $k\in \{(2,0),\dots,(2,s^{(3)}_2-1),(3,0),\dots,(3,s^{(2)}_2-1)\}$.
 \end{enumerate}
 
 \item Consider $l=(t-1,1)$, with $t>2$.
 
 \begin{enumerate}
 
  \item Suppose that following the BMSa we have that $d=2$, $s^{(2)}=(0,2)$, $t= 2s_1^{(1)}$. Then there is a value $b\in\L$ such that the polynomial $$h_b=f^{(1)}-\frac{b}{v}g^{(1)}$$ verifies that $\G_b=\{h_b,f^{(2)}\}$ is a Groebner basis for $\bL(U)$.

\item Suppose that following the BMSa 
we have that $d=3$, $s^{(3)}=(0,2)$, $t=s_1^{(1)}+s_1^{(2)}$. Then 
there is a list of values $c,b,b_0,\dots,b_{s^{(1)}_1-(s^{(2)}_1+1)}$ in $\L$, such that, setting $\mathbf{b}=(b,b_0,\dots,b_{s^{(1)}_1-(s^{(2)}_1+1)})$, the polynomials 
\begin{eqnarray*}
 h_c&=&f^{(2)}-\frac{c}{v_1}g^{(1)}\\
 h_{\mathbf b}&=&f^{(1)}-\frac{b}{v_2}g^{(2)}-\\
 &-&\left(\sum_{n=0}^{s^{(1)}_1-(s^{(2)}_1+1)}b_n X_1^{s^{(1)}_1-(n+s^{(2)}_1+1)}\right)\frac{1}{v_1}g^{(1)}
\end{eqnarray*}
verify that $\G_{(\mathbf b,c)}=\{h_{\mathbf b}, h_c, f^{(3)}\}$ is a Groebner basis for $\bL(U)$.

 \end{enumerate}
 \end{enumerate}
\end{proposition}

\begin{proof}

We will only see the case $l=(1,t-1)$ as the proof the other one is analogous. By hypothesis we may denote $\Delta=\Delta(u^{l+1}) =\Delta(u^l)$.

  \textbf{\textit{1(a).}}  We set $f^{(1)}[U]_l=b$, where it may be, even $b=0$. For $b\neq 0$, by hypothesis, we have that $l-s^{(1)}=(0,t-1)\in \Delta$ and hence it applies Procedure~1 of the BMSa. It is easy to check that, following the notation of Procedure~1 in \cite[Paragraph III.B, pp. 271-272]{BS2}, $i=1$, $j=1$; so that $\mathbf{e}=(0,0)$and  $h_{f^{(1)},g^{(1)}}=f^{(1)}-\frac{b}{v}g^{(1)}=h_b$. If one has $b=0$ then $h_b=f^{(1)}$. Now let us study possible updates for $F_l$. By Theorem~\ref{condiciones suficientes}, independently of any chosen value $b\in\L$, for all $k=(k_1,k_2)\in\S(t)$, such that $l<_{lex}k$, then $k_1\geq 2$; so that, $k-LP(h_b)\not\in\Delta$. As  $|\Delta|=t$, it must be $h_b[U]_k=0$. Then, we do not need to make more tests for $h_b$.

  Now we would have to continue implementing the BMSa to $f^{(2)}$ for all $l<_Tk=(k_1,k_2)\in\S(t)$. By definition of $\S(t)$ it must be $k_2\leq t-2$ (eventually $k_2=0$ is possible), and then $ s^{(2)}\not\preceq k$, so that $f^{(2)}[U]_k = 0$, by convention in Definition~\ref{recurrencia lineal}. So, again, we do not need to make more tests for $f^{(2)}$ and hence $\G_b=\{h_b,f^{(2)}\}$ is a Groebner basis for $\bL(U)$.
  
  \textbf{\textit{1(b).}} We follow Procedure~1 with $i=2$ and $j=1$. Again, $\mathbf{e}=(0,0)$, so that  $h_{f^{(2)},g^{(1)}}=f^{(2)}-\frac{b}{v}g^{(1)}=h_b$ and if one has $b=0$ then $h_b=f^{(2)}$. Now let us study possible updates for $F_{l+1}$. 
  
  As above, by Theorem~\ref{condiciones suficientes}, independently of any chosen value $b\in\L$, for all $k=(k_1,k_2)\in\S(t)$, such that $l<_{lex}k$, then $k_1\geq 2$; so that, $k-LP(h_b)\not\in\Delta$. As  $|\Delta|=t$, it must be $h_b[U]_k=0$. Then, we do not need to make more tests for $h_b$. Now, clearly $k-s^{(1)}\in\Delta$ only if $k=(2,0),\dots,(2,\frac{t}{2}-1),(3,0),\dots,(3,\frac{t}{2}-1)$ and in such places it may be updates, all of them under Procedure~1 of the BMSa; otherwise $k-s^{(1)}\not\in\Delta$ and, as $|\Delta|=t$, it must be $f^{(1)}[U]_k=0$.
  
 \textbf{\textit{1(c).}} We set $f^{(2)}[U]_l=b$ and $f^{(3)}[U]_l=c$, where it may be, even $b=c=0$. For $b,c\neq 0$, by hypothesis, we have that $l-s^{(i)}\in \Delta(u^l)$, for $i=2,3$ and hence, we apply in both cases Procedure~1 of the BMSa. If $b\neq 0$ then following, again, the notation in \cite[Paragraph III.B, pp. 271-272]{BS2} one set $i=2$, $j=2$ then $\mathbf{e}=(0,0)$ and  $h_{f^{(2)},g^{(2)}}=f^{(2)}-\frac{b}{v_2}g^{(2)}=h_b$. If one has $c\neq 0$ then, again under the notation mentioned previously, we set $i=3$, $j=1$ then $\mathbf{e}=(0,0)$ and  $h_{f^{(3)},g^{(1)}}=f^{(3)}-\frac{c}{v_1}g^{(1)}=h_c$. If one has $b=0$ or $c=0$ then $h_b=f^{(2)}$ or $h_c=f^{(3)}$. 
  
  We consider possible updates. Let $l<_T k\in \S(t)$.  First, setting $k=(k_1,k_2)$, since $k_1\geq 2$ we have that, $s^{(3)}\preceq k$ implies $k-s^{(3)}\not\in \Delta $; so that it must happen that $h_c[U]_k=0$ and we do not have to compute such updates. For $k\in \{(2,0),\dots,(2,s^{(3)}_2-1),(3,0),\dots,(3,s^{(2)}_2-1)\}$ it may be updates for $f^{(1)}$, obtaining $\overline f^{(1)}$. 

Consider now $h_b=f^{(2)}-\frac{b}{v_2}g^{(2)}$, constructed above. For those  pairs $k\in \{(2,s^{(2)}_2),\dots,(2,2s^{(2)}_2-1)\}$ it may be updates for $h_b$. Let us study them. For $a=0,\dots,s^{(2)}_2-1$, set $k_a=(2,s^{(2)}_2+a)$.

For $a=0$ and $h_b[U]_{k_0}=b_0$, since $k_0-LP(h_b)\in\Delta$ we have to apply Procedure~1 of the BMSa. In this case, one may check that, under the notation in \cite[Paragraph III.B, pp. 271-272]{BS2} one set $i=2$, $j=1$ then $\mathbf{e}=(0,s^{(2)}_2-1)$ and  $h_{f^{(2)},g^{(1)}}=h_b-\frac{b_0}{v_1}X_2^{s^{(2)}_2-1}g^{(1)}= f^{(2)}-\frac{b}{v_2}g^{(2)}-\frac{b_0}{v_1}X_2^{s^{(2)}_2-1}g^{(1)}=h_{\mathbf b_0}$, where $\mathbf b_0=(b,b_0)$.

Suppose we have constructed, for $0\leq a< s^{(2)}_2-1$ the vector $\mathbf b_{a}=(b,b_0,\dots,b_a)$ and 
\[h_{\mathbf b_a}=f^{(2)}-\frac{b}{v_2}g^{(2)}-\sum_{n=0}^a \frac{b_n}{v_1}X_2^{s^{(2)}_2-(n+1)}g^{(1)} \]
with $LP(h_{\mathbf b_a})=s^{(2)}$, and let $k_{a+1}=(2,s^{(2)}_2+a+1)$. Set $h_{\mathbf b_a}[U]_{k_{a+1}}=b_{a+1}$. If $b_{a+1}\neq 0$ then as $k_{a+1}-LP(h_{\mathbf b_a})\in \Delta$ we have to apply Procedure~1 of the BMSa with $i=2$ and $j=1$; then $\mathbf{e}=(0,s^{(2)}_2-(a+2))$ and  $h_{f^{(2)},g^{(1)}}=h_{\mathbf b_a}-\frac{b_{a+1}}{v_1}X_2^{s^{(2)}_2-(a+2)}g^{(1)}=h_{\mathbf b_{a+1}}$. In case $b_{a+1}=0$ we get $h_{\mathbf b_a}=h_{\mathbf b_{a+1}}$.
 \end{proof}

\begin{proposition}\label{salvando los muebles para los casos del Teorema con Grad}
Take the setting of General Assumption~\ref{Situac general}, with $F_l=\{f^{(1)},\dots,f^{(d)}\}$, and let $g^{(i)}\in G_l$, as in Remark~\ref{Conjunto de los G}, such that $g^{(i)}[u^{k_i}]_{k_i}=v_i\neq 0$, with $k_i<_Tl$, for $i=1,\dots,d-1$. By Remark~\ref{fuera de rango es rec lin}, we may suppose that $\supp(f^{(i)})\setminus LP(f^{(i)})\in \Delta(u^l)$ for $i=1,\dots,d$. Suppose we are implementing the BMSa under the inverse graded ordering ($X_2>X_1$).

\begin{enumerate}
 
 \item Consider $l=(1,t-1)$. 

 \begin{enumerate}

\item Suppose that following the BMSa we have that $d=2$, $s^{(1)}=(2,0)$, $t= 2s_2^{(2)}$. Then there is a value $b\in\L$ such that the polynomial $$h_b=f^{(2)}-\frac{b}{v}g^{(1)}\in \G$$ verifies that:

\begin{enumerate}
 \item If $t=2$ and $\overline f^{(1)}= f^{(1)}-\frac{f^{(1)}[U]_{(3,0)}}{v_1}g^{(1)}$ is the updating of $f^{(1)}$ at the step $k=(3,0)$ of the BMSa then 
 $\G_b=\{\overline f^{(1)}, h_b\}$ is a Groebner basis for $\bL(U)$. 
 
 \item If $t\neq 2$ then $\G_b=\{f^{(1)}, h_b\}$ is a Groebner basis for $\bL(U)$.
\end{enumerate}

\item Suppose that following the BMSa we have that $d=3$, $s^{(1)}=(2,0)$, $t= s_2^{(2)}+s_2^{(3)}$. Then there is a list of values $b,c,c_0,\dots,c_{s^{(3)}_2-(s^{(2)}_2+1)}$ in $\L$, such that, setting $\mathbf{c}=(c,c_0,\dots,c_{s^{(3)}_2-(s^{(2)}_2+1)})$, the polynomials 
\begin{eqnarray*}
 h_b&=&f^{(2)}-\frac{b}{v_2}g^{(2)}\\
 h_{\mathbf c}&=&f^{(3)}-\frac{c}{v_1}g^{(1)}-\\
 &-&\left(\sum_{n=0}^{s^{(3)}_2-(s^{(2)}_2+1)}c_nX_2^{s^{(3)}_2-(n+s^{(2)}_2+1)}\right)\frac{1}{v_2}g^{(2)}
\end{eqnarray*}
verify that $\G_{(b,\mathbf c)}=\{f^{(1)}, h_b, h_{\mathbf c}\}$ is a Groebner basis for $\bL(U)$.

 \end{enumerate}
 
 \item Consider $l=(t-1,1)$, with $t>2$.
 
 \begin{enumerate}
 
 \item Suppose that  $d=2$, $s^{(1)}=(t,0)$ and $s^{(2)}=(0,1)$. Then there is a value $b\in\L$ such that the polynomial $$h_b=f^{(2)}-\frac{b}{v}g^{(1)}$$ 
 verifies the following:  if $\overline f^{(1)}$ is the updating of $f^{(1)}$ at the steps $k=(t,0),\dots,(2t-1,0)$ of the BMSa then 
 $\G_b=\{\overline f^{(1)}, h_b\}$ is a Groebner basis for $\bL(U)$. 
 
  \item Suppose that following the BMSa we have that $d=2$, $s^{(2)}=(0,2)$, $t= 2s_1^{(1)}$. Then there is a value $b\in\L$ such that the polynomial $$h_b=f^{(1)}-\frac{b}{v}g^{(1)}$$ verifies that:
  \begin{enumerate}
    
  \item If $t=2$ it is covered in Item~\textit{1(a)i}.
    \item If $t=4$ and $\overline f^{(2)}= f^{(2)}-\frac{f^{(2)}[U]_{(1,3)}}{v_1}g^{(1)}$ is the updating of $f^{(2)}$ at the step $k=(1,3)$ of the BMSa then  $\G_b=\{h_b,\overline f^{(2)}\}$ is a Groebner basis for $\bL(U)$.
    \item If $t> 4$ then $\G_b=\{h_b,f^{(2)}\}$ is a Groebner basis for $\bL(U)$.
  
  \end{enumerate}

\item Suppose that following the BMSa 
we have that $d=3$, $s^{(3)}=(0,2)$, $t=s_1^{(1)}+s_1^{(2)}$. Then 
there is a list of values $b,b_0,\dots,b_{s^{(1)}_1-(s^{(2)}_1+2)},c,c_0$ in $\L$, such that setting $\mathbf{b}=(b,b_0,\dots,b_{s^{(1)}_1-(s^{(2)}_1+2)})$ and $\mathbf{c}=(c,c_0)$ the polynomials 
\begin{eqnarray*}
 h_{\mathbf b}=f^{(1)}-\frac{b}{v_2}g^{(2)}-\\
 -\left(\sum_{n=0}^{s^{(1)}_1-(s^{(2)}_1+2)}b_n X_1^{s^{(1)}_1-(n+s^{(2)}_1+2)}\right)\frac{1}{v_1}g^{(1)}
\end{eqnarray*}
 and 
 \[h_\mathbf{c}=f^{(2)}-\frac{c}{v_1}g^{(1)}-\frac{c_0}{v_2}g^{(2)}\]
verify that $\G_{(\mathbf b,\mathbf{c})}=\{h_{\mathbf b}, h_{\mathbf{c}}, \overline f^{(3)}\}$ is a Groebner basis for $\bL(U)$,  where $\overline f^{(3)}$ is the final updating of $f^{(3)}$ through the following steps:
\begin{enumerate}
 \item At $(t-2,2)$ in case $s^{(2)}_1\leq 2$ .
 \item At $(t-3,3)$ in case $s^{(3)}_1 \leq 3$ .
 \item At $(0,3)$ in case $t=3$ .
\end{enumerate}
 \end{enumerate}
 \end{enumerate}
\end{proposition}

\begin{proof}
For the sake of completeness, here we will consider $l=(t-1,1)$, which, in fact, include more cases. By hypothesis we may denote $\Delta=\Delta(u^{l+1}) =\Delta(u^l)$. 

\textbf{\textit{2(a).}} We have that $s^{(2)}\preceq l$; so we set $f^{(2)}[U]_l=b$. If $b\neq 0$, then, as $l-s^{(2)}\in \Delta$ we apply Procedure~1 with $i=2$ and $j=1$. Again, $\mathbf{e}=(0,0)$, so that  $h_{f^{(2)},g^{(1)}}=f^{(2)}-\frac{b}{v}g^{(1)}=h_b$ and if one has $b=0$ then $h_b=f^{(2)}$. Let us see possible updatings.

Consider $l<_Tk=(k_1,k_2)\in \S(t)$. If $l<_T k <_T (0,t)$ then, on the one hand $s^{(1)}\not\preceq k$ so that $f^{(1)}[U]_k=0$ by convention in Definition~\ref{recurrencia lineal}; on the other hand, as $k-s^{(2)}\not\in\Delta$ we must have that $h_b[U]_k=0$, because $|\Delta|=t$.

If $(0,t)\leq_Tk$ then, for $k\in \{(0,t),\dots,(0,2t-1)\}$ there are not updates. For $k\in \{(t+1,0),\dots,(2t-1,0)\}$ we have that $LP(h_b)\not\preceq k$; however, $k-s^{(1)}\in \Delta$ so it may be subsequent updates, all of them by applying Procedure~1 of the BMSa, obtaining, finally, $\overline f^{(1)}$ and so  $\G_b=\{\overline f^{(1)}, h_b\}$ is a Groebner basis for $\bL(U)$.

\textbf{\textit{2(b).}} The case $t=2$ has been considered in \textbf{\textit{1(a)i}}, so, we suppose $t>2$. We have that $s^{(2)}\not\preceq l$ and $s^{(1)}\preceq l$; so we set $f^{(1)}[U]_l=b$. If $b\neq 0$ then, as $l-s^{(1)}=(\frac{t}{2}-1,1)\in \
\Delta$ we apply Procedure~1 with $i=1$ and $j=1$. Again, $\mathbf{e}=(0,0)$, so that  $h_{f^{(1)},g^{(1)}}=f^{(1)}-\frac{b}{v}g^{(1)}=h_b$ and if one has $b=0$ then $h_b=f^{(1)}$. Let us see possible updatings.

Consider, again, $l<_Tk=(k_1,k_2)\in \S(t)$. If $l<_T k <_T (0,t)$ then, for $LP(h_b)=s^{(1)}\preceq k$, as $k_2\geq 2$ we must have $(k_1-\frac{t}{2},k_2)\not\in\Delta$, so that $h_b[U]_k=0$, because $|\Delta|=t$. On the other hand, $k-s^{(2)}\in \Delta$ if $k_2=2,3$ and $t-k_2\leq \frac{t}{2}$. Let us study those conditions and recall that $t\geq 4$.

If $k_2=2$ then $t-2\geq \frac{t}{2}$, so that $k-s^{(2)}\not\in \Delta$. If $k_2=3$, $k-s^{(2)}\in \Delta$ implies $t-3<\frac{t}{2}$ so that $t<6$ (we recall that $t$ is even). If $t=4$ then $k-s^{(2)}=(1,1)\in \Delta$ and so it may be updating for $f^{(2)}$ by applying Procedure~1 of the BMSa. In this case, $k=(1,3)$.

Suppose $(0,t)\leq_Tk$. If $k=(k_1,0)$ with $k_1\geq t+1$ then $k-s^{(1)}\not\in\Delta$ and $s^{(2)}\not\preceq k$. If $k=(0,k_2)$ with $k_2\geq t$ then $s^{(1)}\not\preceq k$ and $k-s^{(2)}\not\in \Delta$.

So, for $t=4$ we have that $k=(1,3)$ is the last possible updating; hence $\overline f^{(2)}= f^{(2)}-\frac{f^{(2)}[U]_{(1,3)}}{v_1}g^{(1)}$  form the Groebner basis  $\G_b=\{h_b,\overline f^{(2)}\}$ for $\bL(U)$. As we have seen, if $t>4$ there are no more updatings so that $\G_b=\{h_b,f^{(2)}\}$ is a Groebner basis for $\bL(U)$.

\textbf{\textit{2(c).}} In this case, we have $s^{(1)}\preceq l$, $s^{(2)}\preceq l$ and $s^{(3)}\not\preceq l$. We set $f^{(1)}[U]_l=b$ and $f^{(2)}[U]_l=c$. If $b,c\neq 0$ then, as $l-s^{(1)}\in \Delta$ and $l-s^{(2)}\in\Delta$ we apply Procedure~1 of the BMSa obtaining, $h_b=f^{(1)}-\frac{b}{v_2}g^{(2)}$ and $h_c=f^{(2)}-\frac{c}{v_1}g^{(1)}$. Now we have to see possible updates for $F_{l+1}=\{h_b,h_c,f^{(3)}\}$, with $LP(h_b)=s^{(1)}$ and $LP(h_c)=s^{(2)}$.

As above, we consider  $l<_Tk=(k_1,k_2)\in \S(t)$. Suppose that $l<_T k <_T (t,0)$. For $i=1$, it may happen that $LP(s^{(1)})\preceq k$; however, as $k_2\geq 2$ then $k-s^{(1)}\not\in \Delta$, so there are no updates for $h_b$.  

For $i=2$, if $s^{(2)}\preceq k$ then $k-s^{(2)}\in \Delta$ implies that $k_2-1\leq 1$; as $l<_tk$ then $k_2=2$; so that $k=(t-2,2)$. As $k-s^{(2)}=(s^{(1)}_1-2,1)$ we must have $s^{(1)}_1-2\leq s^{(2)}_1-1$, and then $s^{(1)}_1-s^{(2)}_1=1$. In this case, setting $h_c[U]_{(t-2,2)}=c_0$ (possibly with $c_0=0$), if $c_0\neq 0$, the updating is $ h_{(c,c_0)}= h_c-\frac{c_0}{v_2}g^{(2)}=f^{(2)}-\frac{c}{v_1}g^{(1)}-\frac{c_0}{v_2}g^{(2)}$. If $c_0=0$ then $h_{(c,c_0)}=h_c$.

For $i=3$, we have that $s^{(3)}\preceq k$. Now, if $k-s^{(3)}=(k_1,k_2-2)\in \Delta$ then $k_2-2\leq 1$, so $k_2=2,3$. If $k=(t-2,2)-s^{(3)}\in \Delta$ then $t-2\leq s^{(1)}_1$, hence $s^{(2)}_1\leq 2$. If $k=(t-3,3)-s^{(3)}\in \Delta$ then $t-3\leq s^{(2)}_1$, hence $s^{(1)}_1\leq 3$. In these cases it may be updates for $f^{(3)}$.

Suppose $(0,t)\leq_Tk$. If $k=(0,k_2)$ with $k_2\geq t$ then $s^{(i)}\not\preceq k$, for $i=1,2$ and $s^{(3)}\preceq k$. In the last case, setting $k_2=t+n$, for some $n\in \N$, if $k-s^{(3)}\in \Delta$ then $t+n-2\leq 1$; which means that $n=0$ and $t=3$. This is the last step in which it is possible to make a subsequent updating for $f^{(3)}$.

Suppose that $k=(k_1,0)$, with $k_1\geq t+1$; say $k_1=t+n+1$ with $n\in \N$. For $i=1$, if $k-s^{(1)}\in\Delta$ then $0\leq n\leq s^{(1)}_1-s^{(2)}_1-2$. As in previous arguments, by implementing Procedure~1 of the BMSa, we have that there are (unknown) values $b_0,\dots,b_{
s^{(1)}_1-s^{(2)}_1-2}$ such that, setting $\mathbf{b}=(b,b_0,\dots,b_{s^{(1)}_1-(s^{(2)}_1+2)})$ we may form the polynomial
\begin{eqnarray*}
 h_{\mathbf b}&=&f^{(1)}-\frac{b}{v_2}g^{(2)}-\\
 &-&\left(\sum_{n=0}^{s^{(1)}_1-(s^{(2)}_1+2)}b_n X_1^{s^{(1)}_1-(n+s^{(2)}_1+2)}\right)\frac{1}{v_1}g^{(1)}.
\end{eqnarray*}

For $i=2,3$ we have that $s^{(i)}\not\preceq k$, and this completes the proof.
 \end{proof}

Now we shall consider the cases in which the unknown values are of the form $l=(0,l_2)$ or $l=(l_1,0)$; that is, we deal with the cases \textit{2a)} and \textit{2b)} post the General Assumption~\ref{Situac general}.

\begin{lemma}\label{inferencia en los ejes general}
 Take the setting of General Assumption~\ref{Situac general}. 
 \begin{enumerate}
  \item If $l_1=0$ and $t+s_2^{(d)} \leq l_2$ then $f^{(2)}[U]_l=0$.
  \item If $l_2=0$ and $t+s_1^{(1)}\leq l_1$ then $f^{(1)}[U]_l=0$.
 \end{enumerate}
\end{lemma}
\begin{proof}
In Case~1, we have that $LP(f^{(d)}) \preceq l$. If $f^{(d)}[U]_l\neq 0$ then $t\leq l_2-s_2^{(d)}<l_2-s_2^{(d)}+s_1^{(1)}\leq |\Delta(u^{l+1})|\leq t$, which is impossible.

Case~2 is completely analogous.
\end{proof}

\begin{proposition}\label{inferencia en los extremos de los ejes}
Take the setting of General Assumption~\ref{Situac general}.
 
 \begin{enumerate}
  \item If $l=(0, t+s_2^{(d)}-1)$ then there exists $f\in F_l$ such that $f[U]_l=0$; except in the case $s^{(1)}=(1,0)$ and $d=2$.
  \item If $l=(t+s_1^{(1)}-1,0)$ then there exists $f\in F_l$ such that $f[U]_l=0$; except in the case $s^{(d)}=(0,1)$ and $d=2$.
 \end{enumerate}
\end{proposition}
\begin{proof}
Set $\Delta=\Delta(u^{l+1})$ and let us see the first case, that is $l_2=t+s_2^{(d)}-1$. Then $s^{(d)}\preceq l$. Suppose that $f^{(d)}[U]_l \neq 0$. As $l-s^{(d)}=(0,t-1)$, to have $l-s^{(d)}\in\Delta$ implies that $s^{(1)}=(1,0)$ and $d=2$, because $|\Delta|\leq t$. These are possible values and for them \textbf{we cannot guarantee that we have linear recurring relations}. 

Note that in this case, it may happen that $\Delta(u^l)\varsubsetneq\Delta$ and then we will have to consider Procedure~2 of the BMSa too.

The other case is completely analogous.
\end{proof}

As in results for the case $l_1,l_2\neq 0$, we have again cases in which it is not possible to infer the missing value $u_l$ and we will try an alternative way. 

\begin{proposition}\label{salvando los muebles en los ejes}
Take the setting of General Assumption~\ref{Situac general}. Suppose $d=2$ and $g=g^{(1)}\in G_l$ is such that $g[U]_{k}=v\neq 0$, with $k<_Tl$ (see Remark~\ref{Conjunto de los G}). We also may suppose that, for $i\in \{1,\dots,d\}$, $\supp(f^{(i)})\setminus LP(f^{(i)})\subset \Delta(u^l)$ by Remark~\ref{fuera de rango es rec lin}.
\begin{enumerate}
 \item Suppose that following the BMSa we have that $d=2$, $s^{(1)}=(1,0)$ and $l=(0, t+s_2^{(2)}-1)$. Then, there exists a list $b,b_0,\dots,b_{t-s_2^{(2)}-1}$ of elements of $\L$ such that, setting $\mathbf{b}=(b,b_0,\dots,b_{t-s_2^{(2)}-1})$, the polynomial 
 \begin{eqnarray*}
  h_\mathbf{b}&=&X_2^{t-s^{(2)}_2}f^{(2)}-\frac{b}{v}g-\\
  &&-\left(\sum_{n=0}^{t-s_2^{(2)}-1}X_2^{t-(n+s_2^{(2)}+1)}b_a\right)\frac{1}{b}f^{(2)}
 \end{eqnarray*}
 verifies that $\G_\mathbf{b}=\{\overline f^{(1)}, h_\mathbf{b}\}$ is a Groebner basis for $\bL(U)$, where
 \begin{enumerate}
  \item  $\overline f^{(1)}$ comes from to continue implementing the BMSa \textbf{under the lexicographic ordering} to $f^{(1)}$ at $k\in \{(1,0),\dots,(1,t-1)\}\subset\S(t)$ by applying, in any case, Procedure~1.
  \item $\overline f^{(1)}=f^{(1)}$ if we implement the BMSa under the inverse graded ordering.
 \end{enumerate}
 
 \item  Suppose that following the BMSa we have that $d=2$, $s^{(2)}=(0,1)$ and $l=(t+s_1^{(1)}-1,0)$. Then, there exists a list $b,b_0,\dots,b_{t-s_1^{(1)}-1}$ of elements of $\L$ such that, setting $\mathbf b=(b,b_0,\dots,b_{t-s_1^{(1)}-1})$, the polynomial 
 \begin{eqnarray*}
  h_\mathbf{b}&=&X_1^{t-s^{(1)}_1}f^{(1)}-\frac{b}{v}g-\\
  &&-\left(\sum_{n=0}^{t-s_1^{(1)}-1}X_1^{t-(n+s_1^{(1)}+1)}b_n\right)\frac{1}{b}f^{(1)}
 \end{eqnarray*}
 verifies that $\G_\mathbf{b}=\{h_\mathbf{b}, f^{(2)}\}$ is a Groebner basis for $\bL(U)$.
\end{enumerate}
\end{proposition}

\begin{proof}
We will only see the case in which $l=(0, t+s_2^{(2)}-1)$ as the proof the other are analogous. We denote $\Delta=\Delta(u^{l+1})$ and we set $f^{(2)}[U]_l=b$, where it may be even $b=0$. Supose that $b\neq 0$. Proposition~\ref{inferencia en los ejes general} it must be $d=2$.

In order to make the sets $F_{l+1}$, $G_{l+1}$ and to study subsequent updates we have to consider two cases.

\textbf{Case 1}. Suppose that $s^{(2)}=(0,t)$. Then $l-s^{(2)}\in\Delta(u^l)$ and so, we apply Procedure~1 of the BMSa. Following, again \cite[Paragraph III.B, pp. 271-272]{BS2} one set $i=2$, $j=1$ then $\mathbf{e}=(0,0)$ and  $h_{f^{(2)},g^{(1)}}=f^{(2)}-\frac{b}{v}g=h_{b}$. Then $F_{l+1}=\{f^{(1)},h_b\}$, with $LP(h_b)=s^{(2)}$ and $G_{l+1}=G_l$.

Now we begin by considering possible updates under the lexicographic ordering. If $l<_T k=(k_1,k_2)\in \S(t)$ then $k_1>0$ and hence $k_2<t$ so that $h_b[U]_k=0$, by convention in Definition~\ref{recurrencia lineal}. If $k \in \{(1,0),\dots,(1,t-1)\}$ it may happen that $f^{(1)}[U]_k\neq 0$ and we have then to apply Procedure~1 of the BMSa obtaining some $\overline f^{(1)}$. If $k_1>1$ then it must be $f^{(1)}[U]_k=0$ because $l-s^{(1)}\not\in\Delta$ and $|\Delta|=t$. So, setting $\mathbf b=(b)$ we have that  $\G_\mathbf{b}=\{\overline f^{(1)}, h_\mathbf{b}\}$ is a Groebner basis for $\bL(U)$.

If, otherwise, we are implementing the BMSa under the graded order then $l=(0,2t-1)$ is the last step, and we are finish; so, as $\overline f^{(1)}= f^{(1)}$, setting $\mathbf b=(b)$ we have that  $\G_\mathbf{b}=\{ f^{(1)}, h_\mathbf{b}\}$ is a Groebner basis for $\bL(U)$.
 
 \textbf{Case 2}. We now suppose that $s^{(2)}_2<t$. Then $l-s^{(2)}\not\in\Delta(u^l)$ and so, $\Delta(u^l)\varsubsetneq \Delta$; moreover, $(0,t-1)= l-s^{(2)}\in  \Delta$ so $S=(0,t)$ must be a defining point (see Definition~\ref{descripcion de los Delta}\textit{(2a)} of $\Delta$. Following \cite[Paragraph III.B, pp. 271-272]{BS2} one may see that from all possible new defining points we may only get $S=(0,t)$ which corresponds to Case~4 of Procedure~2, so that our new polynomial will be $h_{b}=h_{f^{(2)},g^{(1)}}$. Then $r=(0,t)$ and $\mathbf e=(0,0)$ so that $h_{b}=X_2^{t-s^{(2)}_2}f^{(2)}-\frac{b}{v}g$. Summarizing, we have $F_{l+1}=\{f^{(1)},h_b\}$ with $LP(h_{b})=(0,t)$ and $G_{l+1}=\{f^{(2)}\}$.
 
 We proceed to study possible updates. First, under the lexicographic ordering, for $l<_T k=(k_1,k_2)\in \S(t)$ we consider $k_1=0$ and $k_1>0$.
 
 If $k_1=0$ then $s^{(1)}\not\preceq k$ and so $f^{(1)}[U]_k=0$ by definition. As $LP(h_b)\preceq k$ we have to see the case with some detail. For $k=(0,t+s^{(2)})$ we set $h_{b}[U]_k=b_0$. If $b_0\neq 0$, as $k-LP(h_{b})\in\Delta$  we apply Procedure~1 of the BMSa with $i=2$ and $j=1$ from which $\mathbf e=(0,t-s^{(2)}_2-1)$ and $h_{(b,b_0)}=h_{b}-\frac{b_0}{b}X_2^{t-(s^{(2)}_2+1)}f^{(2)}$. Note that, if $b_0=0$ we simply get $h_{b}=h_{(b,b_0)}$. By continuing inductively, as in the proof of Proposition~\ref{salvando los muebles para los casos del Teorema con lex}.\textit{.1.c}, we arrive to $F_{(1,0)}=\{f^{(1)},h_{\mathbf b}\}$ and $G_{(1,0)}=\{f^{(2)}\}$.
 
 If $k_1>0$ (that is $(1,0)\preceq k$), then $h_{\mathbf b}[U]_k=0$, because $LP(h_{\mathbf b})\not\preceq k$. Next, we consider $k=(1,0),\dots,(1,t-1)$ together with $s^{(1)}$. In this case $s^{(1)}\preceq k$ and $k-s^{(1)}\in \Delta$ so that any subsequent updating for $f^{(1)}$ must be done by aplying Procedure~1 of the BMSa, obtaining again $s^{(1)}$ as defining point. Let us denote by $\overline f^{(1)}$ the final updating over the mentioned list of pairs. Finally, if $k_1>1$ or $k_2\geq t$ we have that $k-s^{(1)}\not\in\Delta$ so that $\overline f^{(1)}[U]_k=0$, because $|\Delta|=t$. Hence,  $\G_\mathbf{b}=\{\overline f^{(1)}, h_\mathbf{b}\}$ is a Groebner basis for $\bL(U)$.
 
Consider now $l<_T k=(k_1,k_2)\in \S(t)$ under the reverse graded ordering. Then $k_1=0$ or $k_1\geq t+s^{(2)}>0$. Let us begin by considering $k_1>0$. In this case, $k_2=0$, so that $LP(h_\mathbf{b}) \not\preceq k$ and $k-s^{(1)}\not\in \Delta$, so that there are no any possible updating.

If $k_1=0$ then $s^{(1)}\not\preceq k$ and so $f^{(1)}[U]_k=0$ by convention in Definition~\ref{recurrencia lineal}. As $LP(h_b)\preceq k$ we proceed as in the case of the lexicographic ordering. For $k=(0,t+s^{(2)})$ we set $h_{b}[U]_k=b_0$. If $b_0\neq 0$, as $k-LP(h_{b})\in\Delta$  we apply Procedure~1 of the BMSa with $i=2$ and $j=1$ from which $\mathbf e=(0,t-s^{(2)}_2-1)$ and $h_{(b,b_0)}=h_{b}-\frac{b_0}{b}X_2^{t-(s^{(2)}_2+1)}f^{(2)}$. Note that, if $b_0=0$ we simply get $h_{b}=h_{(b,b_0)}$. Again, as in the proof of Proposition~\ref{salvando los muebles para los casos del Teorema con lex}.\textit{.1.c}, we arrive to $F_{(1,0)}=\{f^{(1)},h_{\mathbf b}\}$ and $G_{(1,0)}=\{f^{(2)}\}$.
\end{proof}

\subsection{The final procedure.}\label{final procedure}
Once it has constructed $\G_{\mathbf b}$ or $\G_{\mathbf b,\mathbf c}$, we have to find the concrete values of $\mathbf b$ and $\mathbf c$. To do this, there are not shortcuts.

The only way is the trial and error method, by constructing for each element or vectors of elements of $\L$ the error polynomial, say  $e_{\mathbf b}$ or $e_{\mathbf b,c}$ and checking if one obtains all known values of $U$ (see the following example). As we have a few missing values, by \cite[Corollaries 34, 35]{BS2} we may be sure that there is only one correct value or pair of values.\\

\begin{example}\label{Caso 1 t es 3}
We continue working with Example~\ref{ejemplo planteamiento}; that is $q=2$, $r_1=r_2=15$, so that $\L=\F_{2^4}$. We have an abelian code $C$ in $\F_2(15,15)$ with defining set $\D_{\bs \alpha}(C)=Q(0,0)\cup Q(0,1)\cup Q(0,3)\cup Q(0,5)\cup Q(1,0)\cup Q(3,0)\cup Q(5,0)\cup Q(1,1)\cup Q(2,1)$, with respect to some $\bs\alpha\in \R$. \\
 
 As we comment, by results in \cite{BBCS2}, the strong apparent distance of $C$ is $sd^*(C)=8$, so we set $t=3$.
 
A polynomial $f\in \F_2(15,15)$ was received and we compute the syndrome values over $\D_{\bs \alpha}(C)$. We write, again, for $\tau=(0,0)$ the values corresponding to $\S(3)$.

  \[U=\left(\begin{array}{lllllllllllllll}
   1 & a^6 & a^{12} & a^{9} & a^{9}& 0\\
   
   a^{12} & a^{3} & u_{(1,2)}  \\
   
   a^{9} & a^{8}  \\
   
   a^{7} \\
   
   a^{3}\\
   a^5
  \end{array}\right)\]
  
We recall that the value $u_{(1,2)}$ is missing because $(1,2)\not\in \D_{\bs\alpha}(C)$. Now we implement the BMSa under the lexicographic ordering and we arrive at step $l=(1,2)$ with $F_{(1,2)}= \{X_1+aX_2+a^2,\;X_2^3+a^{6}X_2^2+a^5X_2+a^6\}$, $G_{(1,2)}=\{X_2+a^6\}$, $k_1=(0,3)$ and $(X_2+a^6)[u^{(0,3)}]_{(0,3)}=a$.

As $LP(f^{(2)})\not\preceq (1,2)$ then we have linear recurrence relation by definition.

On the other hand, $LP(f^{(1)})\preceq (1,2)$ and we see that the situation corresponds to Theorem~\ref{inferencia principal}~(1.a). Then we apply Proposition~\ref{salvando los muebles para los casos del Teorema con lex}~(1.a); so that, we set
\[h_b=f^{(1)}+\frac{b}{a}g^{(1)}=X_1+\left(a+\frac{b}{a}\right)X_2+(a^2+a^5b).\]
and so $\G_b=\{h_b,f^{(2)}\}$ will be our Groebner basis.\\

Now we have to find $b\in \L$ such that, following the notation at the beginning of this section, $e_b\left(\bs\alpha^n\right)=u_n$ for all $n\in \S(t)\setminus\{(1,2)\}$.
\begin{description}
 \item[$b=0$.]   Here, $\G_0$ is a reduced Groebner basis for the ideal $\langle \G_0\cup \{X_1^{r_1}-1,X_2^{r_2}-1\}\rangle$, so $e_0=X_1^{11}X_2^{8}+X_1^{4}X_2^{9}+X_1X_2^{4}$. The first element in $\S(t)$, for which equality fails is $(2,1)$, with  $e_0\left(\bs\alpha^{(2,1)}\right)=a^{14}\neq a^8=u_{(2,1)}$.
 \item[$b=a$.]   Here, $\G_a$ is a reduced Groebner basis for the ideal $\langle \G_a\cup \{X_1^{r_1}-1,X_2^{r_2}-1\}\rangle$, so $e_a=X_1^{8}X_2^{9}+X_1^{10}X_2^{8}+X_1^{13}X_2^{4}$. The first element in $\S(t)$, for which equality fails is, again, $(2,1)$, with  $e_a\left(\bs\alpha^{(2,1)}\right)=a^{7}\neq a^8=u_{(2,1)}$.
 \item we continue in this way, until
 \item[$b=a^{11}$.]   Here, $\G_{a^{11}}$ is a reduced Groebner basis for the ideal $\langle \G_{a^{11}}\cup \{X_1^{r_1}-1,X_2^{r_2}-1\}\rangle$, so $e_{a^{11}}=X_1^{14}X_2^{4}+X_1^{2}X_2^{8}+X_1X_2^{9}$. In this case, $e_{a^{11}}\left(\bs\alpha^n\right)=u_n$ for all $n\in \S(t)\setminus\{(1,2)\}$.
\end{description}

Then, we may assert that the error polynomial is $e=X_1^{14}X_2^{4}+X_1^{2}X_2^{8}+X_1X_2^{9}$; so, we decide that the word sent was $e+f$.
\end{example}

\subsection{Examples}

We finish with some examples that showing exceptions in Theorem~\ref{inferencia principal}. In all examples bellow we write, for $l\in \S(t)$, the set $F_l=\{f_l^{(1)},\dots,f_l^{(d_l)}\}$ (recall that $F_{(0,0)}=\{1\}$ by our General Assumption~\ref{Situac general}). 

\begin{example}\label{Caso2b}

 We set $q=2$, $r_1=r_2=15$, so that $\L=\F_{2^4}$. 
For $\tau=(1,0)$ we have
\[\S(4)=\begin{pmatrix}
         1&1&0&1&1&0&1&1\\
         1&0&1&1\\
         1&0&1\\
         1&1\\
         0\\
         1\\
         1\\
         1\\
        \end{pmatrix}\]
and suppose that the missing value is $u_{(3,1)}$.

We arrive at $l=(3,1)$ with $F_{(3,1)}=\{x^2+xy+1,y^2+y+1\}$, $G_{(3,1)}=\{xy+x+1\}$, $k_1=(2,2)$ and $g^{(1)}[u^{k_1}]_{k_1}=1$. In this case, $f_{(3,1)}^{(2)}[U]_{(3,1)}=0$ by convention.

From here, we have the exception appeared in Theorem~\ref{inferencia principal}.\textit{2(b)} in the case of lexicographic ordering. In fact $f_{(3,1)}^{(1)}[U]_{(3,1)}=1$. 
\end{example}

\begin{example}\label{Casos 1c2c}
 We set $q=2$, $r_1=r_2=15$, so that $\L=\F_{2^4}$. 
 
We shall implement the BMSa over
  \[\S(3)=\left(\begin{array}{lllllllllllllll}
   1 & a^2 & a^{4} & a^{6} & a^{8}& a^{10}\\
   
   a^{12} & a^{2} & 1  \\
   
   a^{9} & a^{3}  \\
   
   a^{14} \\
   
   a^{3}\\
   0
  \end{array}\right)\]
under both of the orders considered. 

We consider two cases:
\begin{enumerate}
 \item The missing value is $u_{(1,2)}$.
 \item The missing value is $u_{(2,1)}$.
\end{enumerate}

We begin with the lexicographic ordering. 

We arrive at $l=(1,2)$ with $F_{(1,2)}=\{X_1^2+a^{7}X_1+a^{10}X_2+a^{5},  X_1X_2+a^{3}X_1+a^2X_2+a^{5},                    X_2^2+a^{6}X_2+a^5\}$, $G_{(1,2)}=\{X_2+a^2,X_1+a^{12}\}$,
$k_1=k_2=(1,1)$ and $g^{(i)}[u^{k_i}]_{k_i}=a^{14}$, for $i=1,2$.

Then we have
\begin{enumerate}
 \item $f_{(1,2)}^{(1)}[U]_{(1,2)}=0$ by convention, $f_{(1,2)}^{(2)}[U]_{(1,2)}=a^4$ and $f_{(1,2)}^{(3)}[U]_{(1,2)}=a$. From here, we have the exception appeared in Theorem~\ref{inferencia principal}.\textit{1(c)} in the case of lexicographic ordering.
 \item $f_{(2,1)}^{(1)}[U]_{(2,1)}=1$, $f_{(2,1)}^{(2)}[U]_{(2,1)}=1$ and $f_{(1,2)}^{(3)}[U]_{(1,2)}=0$ by convention. From here, we have the exception appeared in Theorem~\ref{inferencia principal}.\textit{2(c)} in the case of lexicographic ordering.
\end{enumerate}

Now, we implement the BMSa under the inverse graded ordering and we arrive at $l=(2,1)$ with $F_{(2,1)}=\{X_1^2+a^{7}X_1+a^{10}X_2+a^{5},\; X_1X_2+a^{3}X_1+a^2X_2+a^{5},\; X_2^2+a^{6}X_2+a^5\}$, $G_{(1
2,1)}=\{X_2+a^2,\; X_1+a^{12}\}$,
$k_1=k_2=(1,1)$ and $g^{(i)}[u^{k_i}]_{k_i}=a^{13}$, for $i=1,2$.

Then we have

\begin{enumerate}
 \item $f_{(2,1)}^{(1)}[U]_{(2,1)}=1$ and $f_{(2,1)}^{(2)}[U]_{(2,1)}=a^5$. From here, we have the exception appeared in Theorem~\ref{inferencia principal}.\textit{2(c)} in the case of graded ordering.
 \item $f_{(1,2)}^{(2)}[U]_{(1,2)}=a$ and $f_{(1,2)}^{(3)}[U]_{(1,2)}=a$. From here, we have the exception appeared in Theorem~\ref{inferencia principal}.\textit{1(c)} in the case of graded ordering.
\end{enumerate}
\end{example}

%
%
%
%
%
%
%
%
%

\end{document}